 \documentclass[acmsmall,nonacm]{acmart}

\usepackage{booktabs}
\usepackage{bookmark}
\usepackage{hyperref}

\usepackage{wrapfig}

\DeclareMathOperator*{\argmin}{arg\,min}

\newcommand{\Q}{\mathcal{Q}}
\newcommand{\pone}{\textbf{RCQ}}
\newcommand{\E}{\mathcal{E}}
\newcommand{\V}{\mathcal{V}}
\newcommand{\R}{\mathcal{R}}
\newcommand{\T}{\mathcal{T}}

\newcommand{\U}{\mathcal{U}}
\newcommand{\eps}{\varepsilon}
\renewcommand{\epsilon}{\varepsilon}
\renewcommand{\Re}{\mathbb{R}}

\renewcommand{\paragraph}[1]{\smallskip\noindent{\bf {#1. }}}

\newcommand{\allattr}{{\mathbf A}}

\newcommand{\allrel}{{\mathbf R}}

\newcommand{\head}{{\mathbf{y}}}
\renewcommand{\head}{{\mathsf{head}}}

\newcommand{\dom}{{\texttt{\upshape dom}}}
\newcommand{\adom}{{\texttt{\upshape adom}}}

\newcommand{\RCQ}{{\texttt{RCQ}}}
\newcommand{\prob}{\RCQ}
\newcommand{\nnprob}{\texttt{ANNQ}}
\newcommand{\sprob}{{{\texttt{RSQ}}}}
\newcommand{\nn}{\texttt{NN}}

\newcommand{\y}{\mathbf{y}}

\newcommand{\I}{\mathbf{I}}

\newcommand{\dist}{\phi}

\renewcommand{\O}{\tilde{O}}
\newcommand{\rect}{\R}
\newcommand{\algdec}{\textbf{NewD}}
\newcommand{\algheavy}{\textbf{NewH}}
\newcommand{\algnjoin}{\textbf{RangeS}}
\newcommand{\algranget}{\textbf{RangeT}}
\newcommand{\alggist}{\textbf{RangeG}}
\newcommand{\datasyn}{\texttt{synthetic}}
\newcommand{\datataxi}{\texttt{taxi}}
\newcommand{\dataflights}{\texttt{flights}}
\newcommand{\datadblp}{\texttt{dblp}}
\newcommand{\block}{C}
\newcommand{\polylog}{{\textsf{polylog}}}
\renewcommand{\star}{\Q_{\mathsf{star}}^k}
\renewcommand{\path}{\Q_{\mathsf{path}}^k}
\newcommand{\Gpath}{\Q^{\mathsf{path}}}
\newcommand{\primitive}{\mathcal{D}}
\newcommand{\rangetree}{\mathcal{T}}
\newcommand{\hhtw}{\textsf{hhtw}}
\newcommand{\proj}{\bar{\pi}}
\newcommand{\fullversion}[1]{}
\usepackage{multirow}
\allowdisplaybreaks
\usepackage{subcaption}
\usepackage{enumitem}
\setlist{leftmargin=*}

\begin{document}

\title{Space-Time Tradeoffs for Spatial Conjunctive Queries} 

\author{Aryan Esmailpour}
\affiliation{%
  \institution{Department of Computer Science, University of Illinois Chicago}
  \city{Chicago}
  \country{USA}}
\email{aesmai2@uic.edu}

\author{Xiao Hu}
\affiliation{%
  \institution{Cheriton School of Computer Science, University of Waterloo}
  \city{Waterloo}
  \country{Canada}}
\email{xiaohu@uwaterloo.ca}

\author{Stavros Sintos}
\affiliation{%
  \institution{Department of Computer Science, University of Illinois Chicago}
  \city{Chicago}
  \country{USA}}
\email{stavros@uic.edu}
\begin{abstract}
Given a conjunctive query and a database instance, we aim to develop an index that can efficiently answer spatial queries on the results of a conjunctive query. We are interested in some commonly used spatial queries, such as range emptiness, range count, and nearest neighbor queries. These queries have essential applications in data analytics, such as filtering relational data based on attribute ranges and temporal graph analysis for counting graph structures like stars, paths, and cliques. Furthermore, this line of research can accelerate relational algorithms that incorporate spatial queries in their workflow, such as relational clustering. Known approaches either have to spend $\O(N)$ query time or use space as large as the number of query results, which are inefficient or unrealistic to employ in practice. Hence, we aim to construct an index that answers spatial conjunctive queries in both time- and space-efficient ways.  

In this paper, we establish lower bounds on the tradeoff between answering time and space usage. For $k$-star (resp. $k$-path) queries, we show that any index for range emptiness, range counting or nearest neighbor queries with $T$ answering time requires $\Omega\left(N+\frac{N^k}{T^k}\right)$ (resp. $\Omega\left(N+\frac{N^2}{T^{2/(k-1)}}\right)$) space.
Then, we construct optimal indexes for answering range emptiness and range counting problems over $k$-star and $k$-path queries. Extending this result, we build an index for hierarchical queries. By resorting to the generalized hypertree decomposition, we can extend our index to arbitrary conjunctive queries for supporting spatial conjunctive queries.
Finally, we show how our new indexes can be used to improve the running time of known algorithms in the relational setting.
\end{abstract}
\maketitle

\section{Introduction}
In the digital age, tera-bytes of data are generated every second.
A large amount of raw data would be useless if we could not analyze them to extract valuable knowledge. One standard way to extract knowledge is by asking data queries on the input data.
A data analyst might ask several queries to an index built on the entire or a subset of the input data set. 
In relational databases, data is gathered and stored across various tables, where each row corresponds to a tuple and each column to an attribute. Two tuples stored in different tables can be joined if their shared attributes are equal. This setting is quite common in practical database management systems (DBMS) since more than 70\% of DBMS are relational and more than 65\% of the data sets in learning tasks are relational data~\cite{kaggle, link1}.

{\em Spatial conjunctive queries} (SCQ) are spatial queries defined on the results of a {\em conjunctive query} (CQ) over relational data. The most common spatial queries include range emptiness, range counting, and nearest neighbor queries. Due to the importance and ubiquitousness of relational data and spatial queries, SCQs have significant applications in several domains, such as geographic information systems, social networks, and temporal graph analysis.
\begin{example}
\label{ex:1}
    A crypto exchange startup has a table $\mathsf{SellOffer}(\mathsf{sellerID}, \mathsf{CSS}, \mathsf{currency})$ for all selling offers, containing the seller's id, seller's credit score, and the currency the seller wants to sell (bitcoin, Ethereum, etc.).
    It also contains
     a table for all buyers $\mathsf{Buyer}(\mathsf{buyerID}, \mathsf{CSB}, \mathsf{currency})$ that contains the buyer's id, buyer's credit score, and the currency the buyer will use to pay.
    Notice that each seller can make multiple offers, and each buyer might be interested in multiple offers.
    The startup wants to support a Nearest Neighbor classifier on the pairs $(\mathsf{CSS}, \mathsf{CSB})$, where the credit score of a buyer is paired with the credit score of a seller if the buyer and the seller are interested in the same currency.
    The startup aims to construct an index for nearest neighbor queries such that, given a query point $(x,y)$,
    the goal is to find the nearest neighbor in $\pi_{\mathsf{CSS}, \mathsf{CSB}}(\mathsf{SellOffer}\Join \mathsf{Buyer})$),
     where $\pi$ is the projection operation, and $\Join$ is the join operation.  
\end{example}
\begin{example}
\label{ex:2}
Let $G(V,E)$ be a graph representing a social network with $|V|$ users and $N=|E|$ edges representing the friendships between the users. Every user $u$ is associated with a social score, indicating their performance on social media~\cite{linkHookle}.
    A team of researchers aims to count the number of $3$-star user-patterns where the center node has high score and the rest of nodes have low score (nodes $u,v,w,z$ form a $3$-star with center $u$, if $v,w,z$ are all connected to $u$) to understand how the social network influences user interactions.
More formally, let $R(\mathsf{userA},\mathsf{userB})$ be a relation that has a tuple for each edge $(u,v)\in E$. Furthermore, let $S(\mathsf{userA}, \mathsf{score})$ be a relation that has a tuple for each node of the graph along with their social score.
The researchers aim to construct an index such that given a query rectangle $[x_l,x_r]\times[y_l,y_r]\times[z_l,z_r]\times[w_l,w_r]$ (a rectangle contains one interval for every attribute) the goal is to count the number of tuples in
$R(A_1,B)\Join R(A_2,B)\Join R(A_3,B)\Join \sigma_{x_l\leq C_1\leq x_r}(S(A_1,C_1))\Join \sigma_{y_l\leq C_2\leq y_r}(S(A_2,C_2))\Join \sigma_{z_l\leq C_3\leq z_r}(S(A_3,C_3))\Join \sigma_{w_l\leq C_4\leq w_r}(S_4(B,C_4))$,
where $\sigma$ is the selection operation.
\end{example}
\vspace{-0.6em}
Interestingly, in addition to all their real applications, SCQs play a key role in optimization problems over relational data.
For example, Agarwal et al.~\cite{agarwal2025computing}, designed efficient algorithms to find a diverse set on the results of a join (or conjunctive) query in relational data. In their algorithms, they repeatedly use an oracle to count the number of results of a conjunctive query that lie in a query rectangle and an oracle to compute nearest neighbor queries on the results of a conjunctive query.
Furthermore, Esmailpour and Sintos~\cite{esmailpour2025improved} designed efficient algorithms for $k$-means and $k$-median clustering on the results of a join query in relational data, using similar spatial oracles.
In both papers, the authors constructed indexes with $O(N)$ query time to model the required spatial oracles.

There are two main ways to handle SCQs. For range emptiness or range counting queries, given a query hyper-rectangle, one approach~\cite{agarwal2025computing, esmailpour2025improved} is to filter out tuples from the input databases that do not satisfy the linear constraints and then run a standard join algorithm on the remaining tuples. While the space used is only $O(N)$, the query time of this approach is $\Omega(N)$ (such as Yannakakis algorithm~\cite{yannakakis1981algorithms}) to count the number of join results of an acyclic query).
Similarly, for nearest neighbor queries, the indexes constructed in~\cite{deep2022ranked, deep2021ranked, tziavelis2020optimal} can be used as described in~\cite{agarwal2025computing} to compute the nearest neighbor of a query point $q$ in $O(N)$ time for join (or conjunctive) queries using $O(N)$ space.
The other approach is first to materialize all join results and then construct a geometric index (for example, a range tree~\cite{de1997computational} for range counting or quadtree~\cite{finkel1974quad} for approximate nearest neighbor queries) on the join results. While the query time on these indexes is only $\O(1)$ (we use $\tilde{O}$ notation to hide $\polylog(n)$ factors) the space usage is high, which heavily depends on the number of join results. For example, the space needed to materialize all results in Example~\ref{ex:1} is $\Omega(N^2)$, and the space needed in Example~\ref{ex:2} is $\Omega(N^3)$. Overall, no known index exists for SCQs that use less space than materializing all CQ results having sublinear query time with respect to the size of the input database.

\subsection{Conjunctive Queries}
\label{subsec:conj}
Let $\allrel$ be a database schema that contains $m$ relations $R_1, R_2, \cdots, R_m$. Let $\allattr$ be the set of all attributes in $\allrel$. Each relation $R_i$ is defined on a subset of attributes $\allattr_i \subseteq \allattr$. We use $A, B, C, A_1, A_2, A_3, \cdots$ etc. to denote the attributes in $\allattr$ and $a, b, c, \cdots$ etc. to denote their values. Let $\dom(A) = \Re$ be the domain of attribute $A \in \allattr$. The domain of a set of attributes $X \subseteq \allattr$ is defined as $\dom(X) = \Re^{|X|}$.
Given the database schema $\allrel$, let $\I$ be a given database of $\allrel$, and let the corresponding relations of $R_1, \cdots, R_m$ be $R_1^{\I}, \cdots$, $R_m^{\I}$, where $R_i^\I$ is a collection of tuples defined on $\dom(\allattr_i)$. The {\em input size} of database $\I$ is denoted as $N = \sum_{i \in [m]}|R^\I_i|$.
For a database $\I$, the {\em active domain} $\adom^{\I}(A)$ of an attribute $A \in \allattr$ is defined as the set of values from $\dom(A)$ that appear in at least one tuple of $\I$.
We drop the superscripts when $\I$ is clear from the context. 
If an attribute $A\in\allattr$ belongs to more than one relation, it is called \emph{join attribute}. 
%
We consider the class of \emph{conjunctive queries}:
$\Q(\y): -R_1(\mathbf{A}_1), R_2(\mathbf{A}_2), \cdots, R_m(\mathbf{A}_m),$
where $\y \subseteq \allattr$ is the set of {\em output attributes} (a.k.a. {\em free attributes}) and $\mathbf{A} - \y$ is the set of {\em non-output attributes} (a.k.a. {\em existential attributes}). A CQ is {\em full} if $\mathbf{A} = \y$ (a full CQ is also called a \emph{join query}), indicating the natural join among the given relations; otherwise, it is {\em non-full}. 
In this work, we focus on \emph{data complexity}: the query size $m$ and $|\allattr|$ are constants, and the complexity of our algorithms is measured by the input size of the database, $N$.
We can always represent a CQ $\Q$ as a hypergraph $G_{\Q}=(\V,\E)$, where every attribute $A\in\allattr$ is a vertex in $\V$, and every relation $R_j\in\allrel$ is a hyper-edge containing all vertices corresponding to attributes $\allattr_j$.
For an attribute $A\in \allattr$, let $E_A=\{R\in\allrel\mid A\in \allattr\}$.

Extending the notations above, we also use $\head(\Q)$ to denote the output attributes of $\Q$. 
For a subset of attributes $X\subseteq \allattr$ and a set of tuples $Y$, let $\pi_{X}(Y)$ be the set containing the projection of the tuples in $Y$ onto the attributes $X$. Notice that two tuples in $Y$ might have the same projection on $X$; however, $\pi_{X}(Y)$ is defined as a set, so the projected tuple is stored once. We also define the multi-set $\proj_{X}(Y)$ so that if for two different tuples $t_1, t_2\in Y$ it holds that $\pi_{X}(t_1)=\pi_{X}(t_2)$, then the tuple $\pi_{X}(t_1)$ exists more than once in $\proj_{X}(Y)$.
When a CQ $\Q$ is evaluated on database $\I$, its query result denoted as $\Q(\I)$ is the projection of the natural join result of $R_1(\mathbf{A}_1)\Join R_2(\mathbf{A}_2) \Join \cdots \Join R_m(\mathbf{A}_m)$ onto $\head(\Q)$. We study our problem under \emph{bag semantics}, so the result of a conjunctive query might contain duplicates.
More formally, the result of $\Q$ over database $\I$ under bag semantics, is defined as $\Q(\I)=\proj_{\head(\Q)}(\Join_{i\in[m]}R_i(\allattr_i))$.

We notice that bag semantics is the default option in SQL. This choice aligns with the needs of practical database management, where queries often involve aggregations (e.g., COUNT, SUM), and duplicate-preserving behavior is crucial for correct computations. Unlike set semantics, which eliminates duplicates by default, bag semantics enables more efficient query evaluation and aligns with real-world data retrieval scenarios, such as financial transactions, sensor data collection, and reporting tasks. Furthermore, preserving duplicates allows for better optimization strategies in query execution plans, as many relational database engines leverage multi-set operations to improve performance.

We introduce several important classes of CQs that will be commonly used throughout the paper. The $k$-star query is defined as $\star(\mathcal{A}_1,\ldots, \mathcal{A}_k):-R_1(\mathcal{A}_1,\mathcal{B}), R_2(\mathcal{A}_2,\mathcal{B}),\ldots, R_k(\mathcal{A}_k,\mathcal{B}),$
where $\mathcal{A}_i, \mathcal{B}\subseteq \allattr$, $\mathcal{A}_i\cap \mathcal{A}_j=\emptyset$, for $i, j\in[k]$ with $i\neq j$, and $\mathcal{A}_i\cap \mathcal{B}=\emptyset$, for $i\in [k]$. The $k$-path query is defined as $\path(\mathcal{A}_1, \mathcal{A}_2):-R_1(\mathcal{A}_1,\mathcal{B}_1), R_2(\mathcal{B}_1,\mathcal{B}_2),\ldots, R_k(\mathcal{B}_{k-1},\mathcal{A}_2),$
where $\mathcal{A}_1, \mathcal{A}_2, \mathcal{B}_j\subseteq \allattr$, for $j\in[k-1]$, $\mathcal{A}_1\cap \mathcal{A}_2=\emptyset$, and $\mathcal{B}_i\cap \mathcal{B}_j=\emptyset$, for $i,j\in[k-1]$ with $i\neq j$.
We note that for these queries, none of their join attributes belong to the head of $\star$ or $\path$ query. When some join attributes are in the head of $\star$ (resp. $\path$) we call them the generalized $k$-star query (resp. generalized $k$-path query).
A further generalization of star queries is the class of hierarchical queries~\cite{dalvi2007efficient}: A CQ $\Q$ over the set of relations $\allrel$ and set of attributes $\allattr$ is \textit{hierarchical}, if for any pair of attributes $A,B \in \allattr$, either $E_A \subseteq E_B$ or $E_B \subseteq E_A$ or $E_A \cap E_B= \emptyset$.

\subsection{Problem definition}
For a CQ $\Q$, a hyper-rectangle $\rect$ in $\Re^{|\head(\Q)|}$, is defined as $\rect=\delta_1\times\ldots\times\delta_{|\head(\Q)|}$, where $\delta_j=[\delta_j^-, \delta_j^+]$ is an interval denoting the query range predicate on the $j$-th head attribute. We start by defining the $\prob$ problem.
\begin{definition}[Rectangle Counting Query Problem] For a CQ $\Q$ and a database $\I$,  the goal is to construct an index, such that given any hyper-rectangle $\rect=\delta_1\times\ldots\times \delta_{|\head(\Q)|}$, it returns the number of 
query results from $\Q(\I)$ falling into $\rect$, i.e., $|\Q(\I)\cap \rect|=|\sigma_{\delta_j^-\leq A_j\leq \delta_j^+, A_j\in \head(\Q)}(\Q(\I))|$.
\end{definition}
\begin{wrapfigure}{r}{0.4\textwidth}
    \centering
    \includegraphics[scale=0.35]{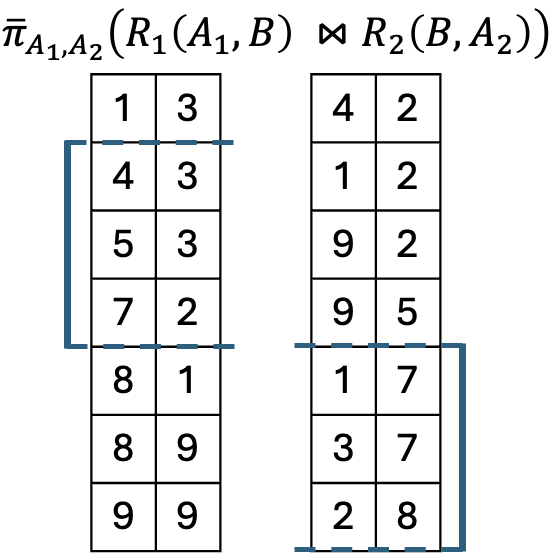}
    \vspace{-0.7em}
    \caption{For query rectangle $\rect=[4,7]\times [7,8]$, $\prob$ returns $3$ because $\bar{\pi}_{A_1,A_2}(\sigma_{4\leq A_1\leq 7}(R_1)\Join \sigma_{7\leq A_2\leq 8}(R_2))=\{(4,7), (5,7), (7,8)\}$. Given the query point $q=(7,2)$ and $\eps=1$, the $\nnprob$ (might) return the query result $p=(9,2)$ because $\dist(q,p)=2\leq (1+1)\cdot 1$, since     the nearest neighbor of $q$ is the tuple $(8,2)$, within distance $1$.}
    \label{fig:probs}
    \vspace{-2.3em}
\end{wrapfigure}
See an example in Figure~\ref{fig:probs}. We use 
$\prob(\Q,\I,\rect)$ to denote the answer to rectangle query $\rect$ over $\Q(\I)$. The range emptiness query can be modeled as an instance of the $\prob$ problem, by checking whether $\prob(\Q,\I,\rect) > 0$. 


Let $\dist: \Re^{|\head(\Q)|} \times \Re^{|\head(\Q)|} \to \Re$ be a distance function. In this work, we focus on the Euclidean distance $\dist(p,q)=\sqrt{\sum_{A \in \head(\Q)}\left(\pi_{A}(p)-\pi_{A}(q)\right)^2},$ for two tuples $p,q\in \Re^{|\head(\Q)|}$.
For any tuple $q\in\Re^{|\head(\Q)|}$, let $\displaystyle{\textsf{dist}(q,\Q(\I))=\min_{p\in\Q(\I)}(p,q)}$ denote the distance between $q$ and the closest result in $\Q(\I)$.
%
%
Let $\nn(\Q,\I,q)$ denote the closest tuple in $\Q(\I)$ from $q$, i.e., 
$\dist(q,\nn(\Q,\I,q))= \dist(q,\Q(\I))$.
Lower bounds on nearest neighbor queries~\cite{liu2004strong, chazelle1996simplex, afshani2012improved} rule out the possibility of a near-linear size index with $\O(1)$ query time even if all tuples belong in one relation.
Hence, we introduce the approximated version of the nearest neighbor query, called the $\nnprob$ problem.

\begin{definition}[$\epsilon$-Nearest Neighbor Query Problem] For a CQ $\Q$ and a database $\I$, the goal is to build an index, such that given an arbitrary point $q \in \Re^{|\head(\Q)|}$ and a constant parameter $\eps\in(0,1]$, it returns a tuple from $\Q(\I)$ that is at most $(1+\eps) \cdot \dist(q,\Q(\I))$ far from $q$.
\end{definition}
We use 
$\nnprob(\Q,\I,q)$ to denote an answer to point query $q$ over $\Q(\I)$, i.e., 
$\dist(q,\nnprob(\Q,\I,q)) \leq (1+\eps)\dist(q,\Q(\I))$. 

\begin{table}[t]
\centering
\scalebox{0.8}{
\begin{tabular}{|c|c|c|c|}
\hline
    \textbf{Query}&\textbf{Space} & \textbf{$\prob$}&\textbf{$\nnprob$}\\\hline
    \multirow{2}{*}{$k$-star} & $\displaystyle{\Omega\left(N+\frac{N^k}{T^k}\right)}$ & \checkmark &\checkmark\\\cline{2-4}
    &$\displaystyle{\O\left(N+\frac{N^k}{T^k}\right)}$ & \checkmark &-\\\cline{1-4}
    \multirow{2}{*}{$k$-path} &$\displaystyle{\Omega\left(N+\frac{N^2}{T^{2/(k-1)}}\right)}$ & \checkmark &\checkmark \\ \cline{2-4}
    & $\displaystyle{\O\left(N+\frac{N^2}{T^{2/(k-1)}}\right)}$ & \checkmark &-\\\cline{1-4}
    Hierarchical & $\displaystyle{\O\left(\min_{\ell\in \{0,\ldots,\mathcal{L}-1\}} \frac{N^m}{T^{(m-1)/(\ell+1)}}+N^{m_{\ell+1}}\right)}$ & \checkmark &\checkmark\\\cline{1-4}
    General & $\displaystyle{\O\left(\min_{\ell\in \{0,\ldots,\mathcal{L}-1\}} \frac{N^{\hhtw\cdot m}}{T^{(m-1)/(\ell+1)}}+N^{\hhtw\cdot m_{\ell+1}}\right)}$ & \checkmark & \checkmark\\ \hline
\end{tabular}
}
    \caption{Summary of our main results for the $\prob$ and $\nnprob$ problems. $T$ is the query time. $N$ is the size of the input database. $m$ is the number of relations in the input query. $\mathcal{L}$ is the depth of the attribute tree of the input CQ $\Q$. $m_{\ell+1}$ is the maximum number of leaf nodes in the subtree rooted at a node of level $\ell+1$ in the attribute tree of $\Q$. $\hhtw$ is the hierarchical hypetree width of $\Q$.} 
    \label{tab:summary}
    \vspace{-2em}
\end{table}
\subsection{Our Results}
In this paper, we propose efficient indexes for spatial conjunctive queries, achieving a better tradeoff between space usage and query time. 
Our main contributions are summarized below (see Table~\ref{tab:summary}).
\begin{itemize}
[leftmargin=*]
    \item We show conditional lower bounds for $\prob$ and $\nnprob$ problems on $k$-star and $k$-path queries, using the \emph{$k$-set-disjointness conjecture} and the \emph{$k$-path-reachability conjecture}, respectively. In particular, we show that any index for the $\prob$ problem (or $\nnprob$ problem) with query time $T$ must use at least $\Omega(N+N^k/T^k)$ space for $k$-star queries,
    (in fact, this lower bound can be extended to any \emph{hierarchical CQ}),
    and $\Omega(N+N^2/T^{2/(k-1)})$ space for $k$-path queries. {\bf (Section~\ref{sec:lb})}
    \item We design near-optimal indexes for the $\prob$ problem on $k$-star and $k$-path queries. In particular, we propose an index for the $\prob$ problem over $k$-star queries with $T$ query time using $\O(N+N^k/T^k)$ space, where $T>0$ is any user-defined parameter. Furthermore, we propose an index for the $\prob$ problem over $k$-path queries with $T$ query time using $\O(N+N^2/T^{2/(k-1)})$ space.
    {\bf (Section~\ref{sec:opt})}
    \item We design indexes for $\prob$ and $\nnprob$ problems on general CQs. We first show an index for $\prob$ and $\nnprob$ problems on generalized $k$-star queries with $T$ query time and $\O(N+N^{k}/T^{k-1})$ space. 
    We extend it to 
    any hierarchical query $\Q$ with $T$ query time and $\displaystyle{\O\left(\min_{\ell\in \{0,\ldots,\mathcal{L}-1\}} \frac{N^m}{T^{(m-1)/(\ell+1)}}+N^{m_{\ell+1}}\right)}$ space, where $\mathcal{L}$ is the depth of the attribute tree of $\Q$, and $m_{\ell+1}$ is the maximum number of leaf nodes in the subtree rooted at a node of level $\ell+1$ in the attribute tree of $\Q$. We then extend this result to arbitrary CQs using hypertree decomposition techniques. 
    {\bf (Section~\ref{sec:general})} 
    \item
    While we mostly focus on the tradeoffs between space usage and query time, for all indexes above we also analyze their preprocessing time. Based on the preprocessing and query time, we show how our new indexes can be used to improve the running time of known relational algorithms for clustering problems over the results of conjunctive queries~\cite{esmailpour2025improved} (relational clustering). 
    In certain cases, our indexes yield a $\sqrt{N}$-factor improvement in the runtime of state-of-the-art relational clustering algorithms. 
    {\bf (Section~\ref{sec:extClust})}. 
    \item 
    Our new indexes are not only of theoretical interest but also practical and easy to implement. We evaluate their performance by implementing them and comparing their query time and space usage against known baseline indexes on real-world and synthetic datasets.
    While known baselines have either super-linear query time or high space consumption, our indexes provide the best space-time tradeoffs. {\bf (Appendix~\ref{sec:experiments})}
\end{itemize}

Even though the space complexity of our index for a hierarchical or general CQ seems high or complicated, we show its strength using a simple example. Consider a full CQ $R_1(A,B,D)\Join R_2(A,B,E)\Join R_3(A,C,F)\Join R_4(A,C,G)$. 
As we discussed above, there are two known indexes for spatial CQs, (for both $\prob$ and $\nnprob$). An index of size $O(N)$ can be used so that each query takes $\O(N)$ time. Alternatively, an index of size $\O(N^{4})$ can be used so that each query takes $\O(1)$ time. Using our result for hierarchical queries, we can construct an index of $\O(N^2)$ space so that each query takes $\O(\sqrt{N})$ time. In some sense, our new indexes balance the space-time tradeoff between the two known extremes.
Overall, we do not claim that our new indexes dominate both the space and the query time of the known indexes. Instead, our new indexes provide a better tradeoff between space and query time. More specifically, our indexes use less space than materializing all join results, and at the same time, they achieve sublinear query time. To the best of our knowledge, no such index was known before for SCQs.

\vspace{-0.5em}
\subsection{Related Work}
Recently, 
space-tradeoff indexes have been studied for reporting, enumeration, or intersection queries on graphs~\cite{agarwal2014space, patrascu2014distance, cohen2010hardness, goldstein2017conditional, xirogiannopoulos2017extracting, agarwal2011approximate} and conjunctive queries in databases
~\cite{deep2023general, zhao2023space, deep2018compressed, kara2020trade, kara2023conjunctive}.
However, none considered spatial queries (such as range counting or approximate nearest neighbor queries) or counting queries so $\prob$ and $\nnprob$ problems cannot be solved using these methods.

Deng et al.~\cite{deng2023space} studied space-query tradeoffs for range subgraph counting and enumeration. Given a graph $G(V,E)$ with $n$ nodes and $m$ edges, each node $v\in V$ is associated with a value $x_v\in \Re$. 
For a general pattern $Q$ of constant size (for example, $k$-star subgraph), they construct an index of $\O(n^2)$ space such that given a query interval $\delta$, it returns the number of $Q$-type subgraphs in $G$ in the range $\delta$. A $Q$-type subgraph $H$ in $G$ is in the range $\delta$ if for every node $v\in H$ it holds $x_v\in\delta$.  They also improve this index for some special types of subgraphs such as $k$-cliques. While their index can handle range counting queries in temporal graphs, and their graph-based indexes can be extended to relational data (similar to our case), their indexes are heavily based on the fact that the query interval $\delta$ is the same across all graph nodes. In our problem, we allow different ranges across different attributes. Their results do not hold in our setting.
Even worse, in Section~\ref{sec:lb}, we show lower bounds that prevent us from having quadratic space index with $\O(1)$ query time for every CQ.

Wang et al.~\cite{wang2023index} studied the
set intersection problem with post filtering. Given a collection of sets $S_1,\ldots, S_n$, the goal is to construct an index such that given two set ids $a, b$ the index answers $F(S_a\cap S_b)$, where $F(\cdot)$ is a filtering function, such as the skyline, the convex hull, and the $k$-nearest items in $S_a\cap S_b$ to a query item $q$. They derive both lower bounds (using the set disjointness conjecture) and design indexes with low space and low query time for common filtering functions $F$. This work does not use ranges over the sets and does not consider relational data, so their methods cannot be directly used to solve our problems.

There is also a lot of work on \emph{similarity join queries}~\cite{agarwal2021dynamic, mamoulis2001multiway, wangoptimal, abo2022complexity} where geometric objects (such as points or rectangles) are stored as tuples in database tables. The goal is to report or count all the joined objects/tuples between two (or more) tables. 
There are different definitions of the similarity join operator. For example, two objects (points) are joined if their distance exceeds a threshold $\tau$, or two objects (rectangles) are joined if they intersect.
These problems do not consider geometric queries on the attributes of the database, so their techniques cannot be used in our problems.




\section{Preliminaries: Geometric indexes}

We mainly need the following two well known oracles to build our indexes. 
\begin{theorem}[\cite{de1997computational}]\label{rect}
Given a finite set of points $P \subset \Re^d$, there exists an index (range tree) with $\O(|P|)$ space, such that, 
given a query hyper-rectangle $\rect$, it returns $|P\cap \rect|$ in $\O(1)$ time.
\end{theorem}
We show the construction of a $d$-dimensional range tree $\rangetree$ on $P\subset \Re^d$ for range counting.
For $d=1$, the range tree on $P$ is a balanced binary search tree $\rangetree$ of $\O(1)$ height. The points of $P$ are stored at the leaves of $\rangetree$ in increasing order, while each internal node $v$ stores the smallest and the largest coordinates, $y_v^-$ and $y_v^+$, respectively, contained in its subtree.
The node $v$ is associated with an interval $I_v=[y_v^-, y_v^+]$ and the subset $P_v=I_v\cap P$.
For $d>1$, $\T$ is constructed recursively: 
We build a $1$D range tree $\rangetree^{(d)}$ on the $x_d$-coordinates of points in $P$. Next, for each node $v\in \rangetree^{(d)}$, we recursively construct a $(d-1)$-dimensional range tree $\rangetree_v$ on $P_v$, which is defined as the projection of $P_v$ onto the hyperplane $x_d=0$, and attach $\rangetree_v$ to $v$ as its secondary tree. The size of $\rangetree$ in $\Re^d$ is $\O(|P|)$ and it can be constructed in $\O(|P|)$ time~\cite{de1997computational}.
For a node $v$ at a level-$\ell$ tree, let $p(v)$ denote its parents in that tree. 
For each node $u$ of the $d$-th level of $\rangetree$, we associate a $d$-tuple $\langle u_1, u_2, \ldots, u_d=u\rangle$, where $u_i$ is the node at the $i$-th level tree of $\rangetree$ to which the level-$(i+1)$ tree containing $u_{i+1}$ is connected.
We associate the rectangle $\square_u=\times_{j=1}^d I_{u_j}$ with node $u$.
For a query rectangle 
$\rect=\times_{i=1}^d\delta_i$, a $d$-level node $u$ is called a \emph{canonical node} if for every $i\in [1,d]$, $I_{u_i}\subseteq \delta_i$ and $I_{p(u_i)}\not\subseteq \delta_i$.
For any rectangle $\rect$, there are $\O(1)$ canonical nodes in $\rangetree$ and can be computed in $\O(1)$ time~\cite{de1997computational}.

For the approximate nearest neighbor problem, we use a \emph{quadtree}~\cite{finkel1974quad} as a black box.
\begin{theorem}[\cite{finkel1974quad, har2011geometric}]\label{nn}
Given a finite set of points $P \subset \Re^d$, there exists an index (quadtree) with $O(|P|)$ space, such that, given a query point $q\in \Re^d$ and a constant parameter $\eps\in(0,1]$, it returns a point $p\in P$ such that $\dist(q,p)\leq (1+\eps)\dist(q,\nn(q,P))$ in $\O(1)$ time.
\end{theorem}
\section{Lower Bounds}
\label{sec:lb}
In this section, we show conditional lower bounds for the \prob~ and \nnprob~ problems on $k$-star and $k$-path queries.
\vspace{-0.2em}

\paragraph{Conjectures}
We build the hardness of our problem with existing hard problems:

\begin{definition}[$k$-set disjointness] Given a universe of integers $\U$, $m$ subsets $\mathcal{S}=\{S_1, S_2, \cdots, S_m\}$ such that $S_i\subseteq \U$ for every $i\in[m]$, with $n=\sum_{i\in[m]}|S_i|$, and an integer $k\geq 1$, 
the goal is to construct an index such that given $k$ integers $j_1, j_2, \dots, j_k \in [m]$ it returns whether $\cap_{h \in [k]}S_{j_h}=\emptyset$.
\end{definition}

\begin{definition}[$k$-path reachability] Given a directed graph $G=(V,E)$ with $n=|E|$, the goal is to construct an index such that given a pair of vertices $u_0, v_0\in V$ it returns whether there exists a path from $u_0$ to $v_0$ consisting of at most $k$ edges.
\end{definition}

 We mainly use the following conjectures as presented in~\cite{goldstein2017conditional}.
\begin{conjecture}\label{conj:inter}[$k$-set disjointness conjecture]
  Any index that answers the $k$-set disjointness problem in time $T$, needs at least $\tilde{\Omega}(n+n^k / T^k)$ space.
\end{conjecture}

\begin{conjecture}\label{conj:reach}[$k$-reachability conjecture]
   Any index that answers the $k$-reachability problem in time $T$, needs at least $\tilde{\Omega}(n+n^2 / T^{\frac{2}{k-1}})$ space. 
\end{conjecture}
\paragraph{Star query}
We consider a simpler version of $\star$ where $\mathcal{A}_i=A_i$ for each $i\in[k]$, and $\mathcal{B}=B$, so $|\mathcal{A}_i|=|\mathcal{B}|=1$. 
\begin{theorem}
\label{lem:lowerbound1}
    Any index answering $\prob$ or $\nnprob$ on $\star$ in $T$ time, uses at least $\tilde{\Omega}(N+N^k / T^k)$ space, assuming the $k$-set disjointness conjecture. 
\end{theorem}

\begin{proof} 
The reductions from $k$-set disjointness problem to $\prob$ or $\nnprob$ are similar; hence we focus on  $\prob$ below.
Given an instance of the $k$-set disjointness problem over the universe $\U$ and $m$ sets $S_1, S_2, \dots, S_m$, we build an instance $\I$ such that using the answer of a \prob~ query over $\star(\I)$, one can answer the set disjointness query in $O(1)$ time. Initially, $\I=\emptyset$. For every pair $a\in\mathcal{U}, S_j$ such that $a\in S_j$, we add the tuple $(j,a)$ in every relation $R_i^D$, for every $i\in [k]$. Consider a $k$-set disjointness query with $k$ distinct integers $j_1,\ldots, j_k$. Let $\rect=[j_1,j_1]\times\ldots\times[j_k,j_k]$ be a rectangle in $\Re^k$.
We next show that the $k$-set disjointness instance is false if and only if $\prob({\star}, \I,\rect)>0$.
    
First, assume that the $k$-set disjointness instance is false, i.e., $\bigcap_{h\in[k]}S_{j_h}\neq\emptyset$. There exists an integer $a\in \U$, such that $a\in \bigcap_{h\in[k]}S_{j_h}$. Hence, for every $h\in[k]$, it holds that $(j_h,a)\in R_{j_h}$. So, the tuple $(j_1,\ldots, j_k)\in \star(\I)$ and by definition $(j_1,\ldots, j_k)\in \rect$. So $\prob({\star}, \I,\rect)$ returns a counter greater than $0$.
    
Second, assume that $\prob({\star}, \I,\rect)>0$, so there exists a tuple $(j_1,\ldots, j_k)\in \star(\I)$. By definition, there should be an integer $a\in \mathcal{U}$ such that $(j_h,a)\in R_j^D$ for every $h\in [k]$. By construction, $a\in \bigcap_{h\in[k]}S_{j_h}$ so the $k$-set disjointness instance is false.
\end{proof}




In Appendix~\ref{appndx:LB}, we show the next lower bound for the $\prob$ and $\nnprob$ problems on path queries.

\begin{theorem}
\label{lem:lowerbound2}
    Any index answering $\prob$ or $\nnprob$ on $\path$ in $T$ time, uses at least $\tilde{\Omega}(N+N^2 / T^{\frac{2}{k-1}})$ space, assuming the $k$-reachability conjecture. 
\end{theorem}

%
\section{Near-Optimal Indexes}\label{sec:opt}
We now propose indexes for the $\prob$ problem on $\star$ and $\path$.
In the end, we show some additional conjunctive queries where $\prob$ and $\nnprob$ can be solved near-optimally.
All missing proofs can be found in Appendix~\ref{appndx:sec4}.
\subsection{Star Queries}
\label{subsec:kstar}
\paragraph{Main idea}
Each relation $R_i$ is first sorted by attribute $A_i$ and partitioned into fixed-size blocks, with each block’s attribute values encapsulated within a minimal bounding interval. A one-dimensional range tree is constructed for each attribute, where leaf nodes store blocks of tuples, and inner nodes store unions of child node blocks. For every combination of tree nodes across attributes, we precompute and store tuple counts corresponding to block intersections. Given a query range, we locate relevant blocks using range trees: Blocks fully contained in the range contribute directly to the count (their count has been precomputed). Partially intersected blocks are processed by iterating over their tuples, and for each such tuple $t$, count the query results using a secondary primitive index that filters tuples based on $\pi_{\mathcal{B}}(t)$.

\paragraph{Primitive index}
Before describing the optimal index, we show the following primitive index to count the number of results from $\star\cap \rect$ whose joined tuples have a fixed projection $b$
 on the attributes $\mathcal{B}$.

    \begin{lemma}\label{lem:domcnt}
    Given $\star$ and a database $\I$ of size $N$, there exists an index $\primitive$ of $\O(N)$ space, such that, given a query tuple $b\in \Re^{|\mathcal{B}|}$ and a hyper-rectangle $\rect\in\Re^{|\head(\star)|}$, the index $\primitive$ has a procedure $\primitive.\mathsf{count}(b,\rect)$ that returns the count $|\{t\in \star(\I)\mid t\in \rect, \forall i\in[k],\exists t_i\in R_i: \pi_{\mathcal{A}_i}(t_i)=\pi_{\mathcal{A}_i}(t), \pi_{\mathcal{B}}(t_i)=b\}|$ in $\O(1)$ time.
    \end{lemma}
In other words, if $\star[b]$ is the query obtained from $\star$ by adding the condition $\mathcal{B}=b$ to its body, 
then $\primitive.\mathsf{count}(b,\rect)=\prob(\star[b],\I,\rect)$.

We start with a simpler case of the problem, where $\star$ contains exactly one attribute $A_i$ in $\mathcal{A}_i$, for $i\in [k]$, and precisely one attribute $B\in \mathcal{B}$.
Lastly, we extend our index to the $k$-star query with any (constant) number of attributes.

\paragraph{Index construction}
 For each $i\in[k]$ we first sort and order all the tuples in $R_i$ based on their value of $A_i$ in ascending order, and store the sorted tuples of $R_i$, in a list denoted by $S_i$. Let $s_{i,j}$ denote the $j$'th tuple in $S_i$ and let $\alpha$ be any integer parameter such that $1
 \leq \alpha\leq N$. We divide each $S_i$, into $\frac{N}{\alpha}$ blocks of size $\alpha$. Let $\block_i$ be the set of blocks in $S_i$ and let $\block_{i,j}$ denote the $j$'th block of $S_i$, i.e., for each $i \in [k]$, tuples $s_{i,[1, \alpha]}$ (the first $\alpha$ tuples of $S_i$) belong to the block $\block_{i,1}$, tuples $s_{i,[\alpha + 1, 2\alpha]}$ belong to the block $\block_{i,2}$, and so on.
 Let $I_{i,j}$ be the smallest interval that contains all values in $\pi_{A_i}(\block_{i,j})$ and let $I_i=\{I_{i,j}\mid j\in[N/\alpha]\}$ be the set of all these smallest intervals.
 For every $i\in[k]$, we build a binary search tree ($1$-dimensional range-tree) $\rangetree_i$ on the intervals $I_i$, i.e., the $j$-th leaf node of the tree corresponds to $I_{i,j}\in I_i$ (we note that intervals in $I_i$ do not intersect in their interiors). If $u$ is the $j$-th leaf node of $\rangetree_i$ then it stores the block $\block_{i}^{(u)}=\block_{i,j}$. If $u$ is an inner node of $\rangetree_i$, then it stores a block that includes the union of all blocks in the leaf nodes of the subtree rooted at $u$. More specifically, if $u$
 has children $v, w$, then we set $\block_{i}^{(u)}=\block_{i}^{(v)}\cup \block_{i}^{(w)}$.
 Every node $u$ is also associated with the corresponding interval $I_i^{(u)}$, which is the smallest interval that contains all values in $\pi_{A_i}(\block_i^{(u)})$.
 For every combination of nodes $u_1\in \rangetree_1,\ldots, u_k\in\rangetree_k$, 
 we compute and store $n_{u_1,\ldots, u_k}=|\{t\in \star(\I)\mid  \forall i\in[k], \pi_{\allattr_i}(t)\in \block_i^{(u_i)}\}|$, 
 i.e., the number of query results whose projection on $\allattr_i=\{A_i,B\}$ is contained in the block $\block_{i}^{(u_i)}$ for every $i\in[k]$.
 Finally, we construct the index $\primitive$ from Lemma~\ref{lem:domcnt} over the query $\star$ and then database instance $\I$.

 \begin{lemma}
 \label{lem:spaceKstar}
     The space of the constructed index is $\O(N+N^k/\alpha^k)$ and its preprocessing time is $\O(N+N^{k+1}/\alpha^{k+1})$.
 \end{lemma}
 
\paragraph{Query procedure}
Given a query rectangle $\rect\in \Re^k$, let $\rect_i=\pi_{A_i}(\rect)=[l_i,r_i]$ be $\rect$'s interval on $A_i$.
We search on $\rangetree_i$ using the range $\rect_i$. Let $\mathcal{I}_i\subseteq \rect_i$ be the largest interval that does not partially intersect any block in the leaf nodes of the tree. Notice that it might be the case that $\mathcal{I}_i=\emptyset$. 
There are $\mu_i=\O(1)$ canonical nodes $U_i=\{u_{(i,1)}, \ldots, u_{(i,\mu_i)}\}$ in $\rangetree_i$ such that 
$\{r\in R_i\mid \pi_{A_i}(r)\in \mathcal{I}_i\}=\bigcup_{h\in[\mu_i]}\block_{i}^{(u_{(i,h)})}$.
Using precomputed values, we compute the following sum over all combinations  $u_{(1,j_1)}\in U_1, \ldots, u_{(k,j_k)}\in U_k$, $\ell_1=\sum_{u_{(1,j_1)}\in U_1, \ldots, u_{(k,j_k)}\in U_k}n_{u_{(1,j_1)},\ldots, u_{(k,j_k)}}$.
By definition, $\ell_1$ computes the number of tuples in $\Q(\I)$ whose projections in $A_1,\ldots, A_k$ lie in blocks completely inside $\rect$.
It remains to handle the blocks partially intersected by $\rect$. For every $i\in[k]$, there are at most $2$ blocks in the leaf level that are partially intersected by $\rect_i$. We visit the relations in order $R_1, R_2,\ldots, R_k$. We first fix the relation $R_1$. Assume that the blocks $\block_{1,j_1}$ and $\block_{1,j'_1}$ are partially intersected by $\rect_1$. For every tuple $t\in \block_{1,j_1}\cup\block_{i,j'_1}$ such that $\pi_{A_1}(t)\in \rect_1$ we set $b=\pi_{B}(t)$ and the rectangle $\rect_t=[\pi_{A_1}(t),\pi_{A_1}(t)]\times\rect_2\times\rect_3\times\ldots\times \rect_k$. We run a query in the primitive index $\primitive$ giving the tuple $b$ and the rectangle $\rect_t$ as the input to the query. We set $x_t=\primitive.\mathsf{count}(b,\rect_t)$.
We compute,
$\ell_2^{(1)}=\sum_{t\in \block_{1,j_1}\cup \block_{1,j'_1}, \pi_{A_1}(t)\in\rect_1}x_t$.
For any other relation $R_i$ with $i\in[2,k]$, we consider the partially intersected blocks, say, $\block_{i,j_i}$ and $\block_{i,j'_i}$. For every tuple $t\in \block_{i,j_i}\cup\block_{i,j'_i}$ such that $\pi_{A_i}(t)\in \rect_i$, we set $b=\pi_B(t)$ and the rectangle $\rect_t=\mathcal{I}_1\times\mathcal{I}_2\times\ldots\times\mathcal{I}_{i-1}\times[\pi_{A_i}(t),\pi_{A_i}(t)]\times \rect_{i+1}\times \rect_k$.
We run a query in the primitive index $\primitive$, giving the tuple $b$ and the rectangle $\rect_t$ as the input to the query. We set $x_t=\primitive.\mathsf{count}(b,\rect_t)$.
We compute,
$\ell_2^{(i)}=\sum_{t\in \block_{i,j_i}\cup \block_{i,j'_i}, \pi_{A_i}(t)\in\rect_i}x_t$.
\noindent In the end, we return $\ell_1+\sum_{i\in[k]}\ell_2^{(i)}$.
We show the correctness of our index in Appendix~\ref{appndx:sec4}.

    \paragraph{Query time}
   The canonical nodes $U_i$
 are derived in $\O(1)$ time, by the query procedure of range trees. Since there are $\O(1)$ canonical nodes in every tree $\rangetree_i$, all combinations of canonical nodes over the $k$ different trees are $\O(1)$, so the value $\ell_1$ is computed in $\O(1)$ time. Next, every value $\ell_2^{(i)}$ is computed by running $O(\alpha)$ queries to the primitive index $\primitive$. Every query on $\primitive$ takes $\O(1)$ time, so in total the query procedure takes $\O(k\cdot\alpha)=\O(\alpha)$ time.

In Appendix~\ref{appndx:sec4}, we straightforwardly  extend our index to handle star queries where $\mathcal{A}_i$, $\mathcal{B}$ contain more than one attributes. 
We conclude with the next theorem.

\begin{theorem}\label{alg:star}
    Given a star CQ $\star$ on a database $\I$ of size $N$, and a parameter $1\leq T\leq N$, there exists an index of $\O(N+N^{k}/T^{k})$ space that is constructed in $\O(N+N^{k+1}/T^{k+1})$ time, such that given a hyper-rectangle $\rect\in \Re^{|\head(\star)|}$, it returns $\prob(\star,\I,\rect)$ in $T$ query time.
\end{theorem}


\subsection{Path queries}
\label{subsec:path}
\paragraph{Main idea} 
We construct an index for efficiently processing path queries by leveraging hierarchical data partitioning and precomputed aggregations, similarly to star queries. The index is built inductively: for $k=2$, the problem reduces to a star query, and for general $k$, we assume an index exists for $(k-1)$-path queries. The tuples in the first and last relations are sorted based on their respective attributes and divided into fixed-size blocks, encapsulated within minimal bounding intervals. One-dimensional range trees are built over these intervals, where leaf nodes store blocks, and inner nodes store unions of child node blocks. For every combination of tree nodes across attributes, we precompute and store tuple counts corresponding to block intersections.
Given a query range, we locate relevant blocks using the range trees: blocks fully contained in the range contribute directly to the count (their count has been precomputed).
Partially intersected blocks are resolved by iterating over their tuples, and for each such tuple $t$ in $R_1$ (resp. $R_k$), count the query results given that $t$ is the tuple in $R_1$ (resp $R_k$) by querying the $(k-1)$-path index.

We start with a simpler instance of the problem, where $\path$ contains exactly one attribute $A_1\in \mathcal{A}_1$, $A_2\in \mathcal{A}_2$, and $B_i\in\mathcal{B}_i$ for $i\in [k-1]$. Hence, $\path(A_1, A_2):-R_1(A_1,B_1),\ldots, R_k(B_k,A_2)$. In the end we extend our index to the general path query.

\newcommand{\ds}{\mathcal{D}}
\paragraph{Index construction}
We show an inductive construction of our new index. For $k=2$, the star query is equivalent to the $2$-path query so let $\ds_2$ be the index from Theorem~\ref{alg:star}.
Without loss of generality, assume that the index for the $\prob$ problem on any $(k-1)$-path query is known and denoted by $\ds_{k-1}$. We show the construction of an index $\ds_k$ for the $\prob$-problem on $\path$.


 We first sort and order all the tuples in $R_1$ (resp. $R_k$) based on their values of $A_1$ (resp. $A_{2}$) in ascending order, and store the sorted tuples of $R_1$ (resp. $R_k$), in a list denoted by $S_1$ (resp. $S_k$).
 Let $\alpha$ be any integer parameter between $1$ and $N$. We divide $S_1$ (resp. $S_k$), into $\frac{N}{\alpha}$ blocks of size $\alpha$. Let $\block_1$ (resp. $\block_k$) be the set of blocks in $S_1$ (resp. $S_k$) and let $\block_{1,j}$ (resp. $\block_{k,j}$) denote the $j$-th block of $S_1$ (resp. $S_k$).
 Let $I_{1,j}$ (resp. $I_{2,j}$) be the smallest interval that contains all values in $\pi_{A_1}(\block_{1,j})$ (resp. $\pi_{A_2}(\block_{k,j})$) and let $I_1$ (resp. $I_2$) be the set of these intervals in $S_1$ (resp. $S_k$).
 We build a binary search tree ($1$-dimensional range-tree) $\rangetree_1$ (resp. $\rangetree_2$) on the intervals $I_1$ (resp. $I_2$). If $u$ is the $j$-th leaf node of $\rangetree_1$ (resp. $\rangetree_2$) then it stores the block $\block_{1}^{(u)}=\block_{1,j}$ (resp. $\block_{k}^{(u)}=\block_{k,j}$). If $u$ is an inner node of $\rangetree_1$ (resp. $\rangetree_2$), then it stores a block that includes the union of all blocks in the leaf nodes of the subtree rooted at $u$. More specifically, if $u$
 has children $v, w$, then we set $\block_{1}^{(u)}=\block_{1}^{(v)}\cup \block_{1}^{(w)}$ (resp. $\block_{k}^{(u)}=\block_{k}^{(v)}\cup \block_{k}^{(w)}$).
 Every node $u$ is also associated with the corresponding interval $I_1^{(u)}$ (resp. $I_2^{(u)}$), which is the smallest interval that contains all values in $\pi_{A_1}(\block_1^{(u)})$ (resp. $\pi_{A_{2}}(\block_k^{(u)})$).
 For every combination of nodes/blocks $u_1\in \rangetree_1, u_2\in\rangetree_2$, we compute and store $n_{u_1,u_2}=|\{t\in \path(\I)\mid 
 \pi_{\allattr_1}(t)\in \block_1^{(u_1)}, \pi_{\allattr_k}(t)\in \block_k^{(u_2)}\}|$, 
 i.e., the number of query results whose projection on $\allattr_1=\{A_1,B_1\}$ is contained in the block $\block_{1}^{(u_1)}$ and whose projection on $\allattr_{k}=\{B_k,A_2\}$ is contained in the block $\block_{k}^{(u_2)}$.
 Finally, we construct the index $\ds_{k-1}^{(1)}$ on $\Gpath_{k-1}(A_1,B_k):-R_1(A_1,B_1),\ldots, R_{k-1}(B_{k-1}, B_k)$ and the index $\ds_{k-1}^{(2)}$ on $\Gpath_{k-1}(B_1,A_{2}):-R_2(B_1,B_2),\ldots, R_{k}(B_{k}, A_{2})$.

 \begin{lemma}
 \label{lem:spaceKpath}
     The space of the constructed index is $\O(N+N^2/\alpha^2)$ and its preprocessing time is $\O(N+N^3/\alpha^3)$.
 \end{lemma}
 
\vspace{-0.5em}
\paragraph{Query procedure}
Given a query rectangle $\rect\in \Re^2$, let $\rect_1=\pi_{A_1}(\rect)$ (resp. $\rect_2=\pi_{A_{2}}(\rect)$) be the query interval on $A_1$ (resp. $A_{2}$).
We search on $\rangetree_1$ using the range $\rect_1$ and $\rangetree_2$ using the range $\rect_2$. Let $\mathcal{I}_1\subseteq \rect_1$ (resp. $\mathcal{I}_2\subseteq \rect_2$) be the largest interval that does not partially intersect any block in the leaf nodes of the tree.
Notice that it might $\mathcal{I}_1=\emptyset$ or $\mathcal{I}_2=\emptyset$. 
There are $\mu_1=\O(1)$ (resp. $\mu_2=\O(1)$) canonical nodes $U_1=\{u_{(1,1)}, \ldots, u_{(1,\mu_1)}\}\}$ in $\rangetree_1$ (resp. $U_2=\{u_{(2,1)}, \ldots, u_{(2,\mu_2)}\}$ in $\rangetree_2$) such that $\{r\in R_1\mid \pi_{A_1}(r)\in\mathcal{I}_1\}=\bigcup_{h\in[\mu_1]}\block_{1}^{(u_{(1,h)})}$, resp. $\{r\in R_k\mid \pi_{A_2}(r)\in\mathcal{I}_2\}=\bigcup_{h\in[\mu_2]}\block_{k}^{(u_{(k,h)})}$.
Using the precomputed values, we compute the following sum over all combinations $u_{(1,j_1)}\in U_1, u_{(2,j_2)}\in U_2$, 
$\ell_1=\sum_{u_{(1,j_1)}\in U_1, u_{(2,j_2)}\in U_2}n_{u_{(1,j_1)}, u_{(2,j_2)}}$.
By definition, $\ell_1$ computes the number of tuples in $\Q(\I)$ whose projections in $A_1, A_2$ lie in blocks completely inside $\rect$.

It remains to handle the blocks that are partially intersected by $\rect$. There are at most $2$ blocks/leaf nodes in $\rangetree_1$ (resp. $\rangetree_2$) that are partially intersected by $\rect_1$ (resp. $\rect_2$).
Assume that the blocks $\block_{1,j_1}$ and $\block_{1,j'_1}$ are partially intersected by $\rect_1$ and the blocks $\block_{k,j_k}$ and $\block_{k,j'_k}$ are partially intersected by $\rect_2$.
For every tuple $t\in \block_{1,j_1}\cup \block_{1,j'_1}$ such that $\pi_{A_1}(t)\in\rect_1$, we set the rectangle $\rect_t=[\pi_{B_1}(t),\pi_{B_1}(t)]\times \rect_2$. We run a query on $\ds_{k-1}^{(2)}$ using the range $\rect_t$. Let $\ds_{k-1}^{(2)}(\rect_t)$ be the result of this query.
Next, for every tuple $t\in \block_{k,j_k}\cup \block_{k,j'_k}$ such that $\pi_{A_{2}}(t)\in\rect_2$, we set the rectangle $\rect_t'=\mathcal{I}_1\times [\pi_{B_k}(t), \pi_{B_k}(t)]$. We run a query on $\ds_{k-1}^{(1)}$ using the range $\rect_t'$. Let $\ds_{k-1}^{(1)}(\rect_t')$ be the result of this query.
We compute,
$\ell_2=\sum_{t\in\block_{1,j_1}\cup \block_{1,j'+1}, \pi_{A_1}(t)\in\rect_1}\ds_{k-1}^{(2)}(\rect_t)+\sum_{t\in\block_{k,j_k}\cup \block_{k,j'_k}, \pi_{A_{2}}(t)\in\rect_2}\ds_{k-1}^{(1)}(\rect_t')$
\noindent In the end, we return $\ell_1+\ell_2$.
We show the correctness of our index in Appendix~\ref{appndx:sec4}.

\paragraph{Query time}
Let $\mathcal{Y}_k$ be the query time of our index. The value $\ell_1$ is computed in $\O(1)$ since $|U_1|, |U_2|=\O(1)$. Furthermore, by the construction of blocks, it holds that $|\block_{1,j_1}\cup \block_{1,j'_1}|\leq 2\alpha$ and $|\block_{k,j_k}\cup \block_{k,j'_k}|\leq 2\alpha$. Hence, we run at most $4\alpha$ queries on $\ds_{k-1}^{(1)}$ and $\ds_{k-1}^{(2)}$. So, the query time is $\mathcal{Y}_k\leq 4\alpha\mathcal{Y}_{k-1}+\O(1)$. By definition, $\mathcal{Y}_2=\O(\alpha)$ and $k$ is a constant so $\mathcal{Y}_k=\O(4^{k-1}\alpha^{k-1})=\O(\alpha^{k-1})$. 

Overall, we constructed an index for the $\prob$ problem on $k$-path queries, with query time $\O(\alpha^{k-1})$, space $\O(N+N^2/\alpha^2)$, and preprocessing time $\O(N+N^3/\alpha^3)$. Equivalently, if we set $\beta=\alpha^{k-1}$, our index has query time $\O(\beta)$, space $\O(N+N^2/\beta^{\frac{2}{k-1}})$, and preprocessing time $\O(N+N^3/\beta^{\frac{3}{k-1}})$.
The extension to $k$-path queries having a set of attributes $\mathcal{B}_i$ instead of one attribute $B_i$ (and $\mathcal{A}_1, \mathcal{A}_2$ instead of $A_1, A_2$, respectively) is straightforward, as in Subsection~\ref{subsec:kstar}.
    \begin{theorem}\label{alg:path}
    Given $\path$ on a database $\I$ of size $N$, and a parameter $1\leq T\leq N$, there exists an index of $\O(N+N^{2}/T^{\frac{2}{k-1}})$ space that is constructed in $\O(N+N^{3}/T^{\frac{3}{k-1}})$ time, such that given a hyper-rectangle $\rect\in \Re^{|\head(\star)|}$, it returns $\prob(\path,\I, \rect)$ in $T$ query time.
\end{theorem}


\fullversion{
\subsection{Additional optimal queries}
A trivial lower bound for the $\prob$ and $\nnprob$ problems is an index with $O(n)$ space and $\O(1)$ query time. In this subsection, we show some conjunctive queries where such an optimal index can be constructed.

\begin{theorem}
    Given a CQ $\Q$ on a database $\I$ of size $N$, if there exists $j\in[m]$ such that $\head(\Q)\subseteq \allattr_j$, then there exists an index with $\O(N)$ space and $\O(1)$ query time for the $\prob$ and the approximate  $\nnprob$ problem.
\end{theorem}
\begin{proof}
    For any tuple $t\in\proj_{\head(\Q)}(R_j)$ we compute $w_t=|\{t'\in\Q(\I)\mid \pi_{\head(Q)}(t')=t\}|$. We construct a $|\head(\Q)|$-dimensional range tree for the weighted version of range counting queries, over $\proj_{\head(\Q)}(R_j)$. The index has $\O(|R_j|)=\O(N)$ space. Given a query rectangle $\rect$ we run a (weighted) range counting query in $\O(1)$ time. Similarly, for the $\nnprob$, after keeping only the non-dangling tuples in the database we construct an approximate nearest neighbor index (quadtree) over the points in $\proj_{\head(\Q)}(R_j)$. The result follows.
\end{proof}
}

\section{General Indexes}\label{sec:general}
We develop indexes for range counting or approximate nearest neighbor queries on the results of any conjunctive query $\Q$.
Before we continue, we argue that $\prob$ 
on conjunctive queries is equivalent to $\prob$ 
on full conjunctive queries (join queries) under bag semantics. Hence, for simplicity, in the next subsections, we present the indexes and the analysis for join queries. All missing proofs can be found in Appendix~\ref{appndx:sec5}.
\begin{lemma}
\label{lem:equiv}
Let $\I$ be a database instance.
Assume an index for the $\prob$ problem over a full CQ $\Q'(\bigcup_{i\in [m]}\mathbf{A}_i):-R_1(\mathbf{A}_1), R_2(\mathbf{A}_2), \cdots, R_m(\mathbf{A}_m)$ with query time $T$ space $S$, and preprocessing time $P$. The same index, with space $S$ and preprocessing time $P$, can be used over $\Q(\mathbf{y}):-R_1(\mathbf{A}_1), R_2(\mathbf{A}_2), \cdots, R_m(\mathbf{A}_m)$, with $\mathbf{y}\subset \bigcup_{i\in [m]}\mathbf{A}_i$, such that given a query rectangle $\rect\in\Re^{|\mathbf{y}|}$,
the value of $\prob(\Q,\I,\rect)$ is computed in $T$ time.

\end{lemma}
Recall that the generalized $k$-star CQ query is defined as $\Q(\y):-R_1(\mathcal{A}_1,\mathcal{B}), \ldots, R_k(\mathcal{A}_k,\mathcal{B})$, where $\y$ is any subset of $\mathcal{B}\bigcup\left(\bigcup_{i\in[k]}\mathcal{A}_i\right)$.
In the next subsection, we first show an index for the $\prob$ problem on $k$-star join queries. By Lemma~\ref{lem:equiv}, it equivalently handles generalized $k$-star CQs. Then we extend this index to any hierarchical query. Finally, we use the hierarchical generalized hypertree decomposition to develop an index for any conjunctive query, under bag semantics.
In all cases, we extend the results to the $\nnprob$ problem.

\subsection{Generalized $k$-Star Queries}
\label{subsec:Gkstar}
%
We describe an index for answering $\prob$ and $\nnprob$ on $\star(\mathcal{A}_1,\ldots, \mathcal{A}_k, \mathcal{B}):-R_1(\mathcal{A}_1,\mathcal{B}),\ldots, R_k(\mathcal{A}_k,\mathcal{B})$ in time $T$ using $\O(N^{k}/T^{k-1})$ space.

\paragraph{Main idea}
Our index construction partitions values of attribute(s) $\mathcal{B}$ into two categories: \emph{heavy} and \emph{light}, based on the number of tuples associated with each value.
We precompute and store all join results where the projection on $\mathcal{B}$ belongs to the light set, constructing a range counting index $D^L$ over these results. For each heavy value $b_i$, we construct separate range counting indexes $D^H_{i,j}$ for each relation $R_j$. Given a query rectangle, we first use $D^L$ to count results with light projections. Then, for each heavy value, we compute the number of tuples in each relation within the query range using $D^H_{i,j}$ and aggregate these counts across relations. The final result is the sum of both counts.
Intuitively, light values generate only a small number of join results (so $D^L$ has small size) and the heavy values are limited (making their traversal efficient).

\paragraph{Index construction}
    Let $\alpha$ be any integer parameter between $1$ and $N$, and let $R_j^{b_i} = \{t \in R_j | \pi_{\mathcal{B}}(t) = b_i \}$ be the set of tuples in relation $R_j$ whose projection on $\mathcal{B}$ is $b_i$, for each $b_i \in \adom(\mathcal{B})$ and each $j \in [k]$. Let $T_i = \bigcup_{j \in [k]} R_j^{b_i}$. We separate $b_i \in \adom(\mathcal{B})$ into two different sets \emph{heavy} and \emph{light}. Let $H = \{b_i \in \adom(\mathcal{B})| \frac{N}{\alpha} \leq |T_i|\}$ be the set of heavy $b_i$'s and let $L = \{b_i \in \adom(\mathcal{B})| \frac{N}{\alpha} > |T_i|\}$ be the set of light $b_i$'s. We construct all query results having a light projection on the $\mathcal{B}$ attribute(s) i.e., $Q_L = \{t \in \star(\I) | \pi_{\mathcal{B}}(t) \in L\}=\bigcup\limits_{b_i \in L}(\Join_{j \in [k]} R_j^{b_i})$.
    Then we construct a range counting index $D^L$ (from Theorem \ref{rect}) over the points in $Q_L$. Next, we go through the heavy  values one by one and for each $b_i \in H$, compute $R_j^{b_i}$ for each $j\in [k]$. For each $b_i \in H$, and each $j \in [k]$, we construct a separate range counting index $D^H_{i,j}$ (from Theorem \ref{rect}) over the points in $R_j^{b_i}$.

  \paragraph{Query procedure}
  We are given a query rectangle $\rect$. For each $j\in[k]$, let $\pi_{\allattr_j}(\rect)$ be the rectangle defined by the linear predicates with respect to the attributes in $\allattr_j$.
  Using the range tree $D^L$, we compute $n_1=D^L(\rect)$ which is the number of join results in the query rectangle with light projections on $\mathcal{B}$.
Next, compute the number of join results in the query rectangle with heavy projections on $\mathcal{B}$. For each $b_i \in H$ and $j\in[k]$, we use the range tree $D^H_{i,j}$ to compute $h_{i,j}=D^H_{i,j}(\pi_{\allattr_j}(\rect))$, i.e., the number of tuples in $R_j$ in rectangle $\pi_{\allattr_i}(\rect)$ whose projection on $\mathcal{B}$ is $b_i$.
Let $n_2 = \sum_{b_i \in H}(\prod_{j \in [k]}h_{i,j})$.
We return $n_1+n_2$.
Due to page limits, we present the correctness proof, as well as the query time and space analysis of our index, in Appendix~\ref{appndx:sec5}.
\fullversion{
    \paragraph{Correctness}
    $D^L(\rect)$ counts every tuple $t\in\star(\mathcal{A}_1,\ldots,\mathcal{A}_k,\mathcal{B})\cap \rect$ with light projection on $\mathcal{B}$, i.e., $\pi_{\mathcal{B}}(t)\in L$. Similarly, by definition, $\sum_{b_i \in H}(\prod_{j \in [k]}h_{i,j})$ counts all tuples $t\in\star(\mathcal{A}_1,\ldots,\mathcal{A}_k,\mathcal{B})\cap \rect$ such that $\pi_{\mathcal{B}}(t)=b_i\in H$. The result follows.
    }


\begin{theorem}\label{range}
    Given a generalized $k$-star CQ $\Q$ on a database $\I$ of size $N$, and a parameter $1\leq T\leq N$, there exists an index of $\O(N+N^{k}/T^{k - 1})$ space and the same preprocessing time,  such that given a hyper-rectangle $\rect\in \Re^{|\head(\Q)|}$, it returns $\prob(\Q,\I,\rect)$ in $T$ query time.
\end{theorem}

In Appendix~\ref{appndx:annExt}, we show how to extend the index to the $\nnprob$ problem.
\begin{theorem}
\label{thm:nnStar}
      Given a generalized $k$-star CQ $\Q$ on a database $\I$ of size $N$, and a parameter $1\leq T\leq N$, there exists an index of $\O(N+N^{k}/T^{k - 1})$ space and the same preprocessing time,  such that given a query point $q\in \Re^{|\head(\Q)|}$, it returns $\nnprob(\Q,\I,q)$ in $T$ query time.
\end{theorem}

\subsection{Hierarchical Queries}
\label{subsec:hierarchical}
In this section, we study the $\prob$ and $\nnprob$ problems on hierarchical queries.  
Hierarchical queries can be thought of as a generalization of the star queries. To answer $\pone$ queries, we build the attribute tree that always exists for hierarchical queries.
\vspace{-0.3em}
\begin{definition}[Attribute Tree]
    The \textit{attribute tree} of a hierarchical CQ $\Q$ over the set of relations $\allrel$ and the attributes $\allattr$, denoted by $\T(\Q)$ is a rooted tree such that each node corresponds to an attribute in $\allattr$; $x$ is a descendant of $y$ if $E_x \subseteq E_y$; 
    and each path from the root to a leaf node corresponds to a relation in $\allrel$.  
\end{definition}
\vspace{-0.3em}
If $\T(\Q)$ is not connected, it is easy to see that we can solve $\pone$ on each connected component and then multiply all the answers to get an answer for the whole query. Thus, we only focus on the case where $\T(\Q)$ is connected and propose the following index for solving $\pone$ on full hierarchical queries. 

\paragraph{Main idea}
Our index construction follows a similar strategy to the previous method, where values of a key attribute are categorized as \emph{heavy} or \emph{light} to optimize query performance. However, instead of partitioning the values based on attributes $\mathcal{B}$, we partition the the tuples of the root node $r$ in the attribute tree as light and heavy.
We precompute all join results where the projection on $r$ belongs to the light set, constructing a range counting index $D^L$ over these results.
For a heavy value $r_i$
instead of indexing each relation separately as in the previous approach, we group relations based on shared child attributes and construct range counting indexes $D^H_{i,c}$ for each child attribute $c$. Given a query rectangle, we first retrieve results for light values using $D^L$. For each heavy value, we compute relevant tuple counts for each child attribute and aggregate them across all children.

\paragraph{Index construction}
Let $r$ be the root attribute in $\T(\Q)$ and let $C = \{c_1, \dots, c_p \}$ be the children of the root that we call the \textit{child} attributes. It is easy to see that the root attribute $r$ appears in all of the $k$ relations in $\allrel$. Let $\alpha$ be any integer parameter between $1$ and $N$. For each $r_i \in \adom(r)$ and each $j \in [m]$, let $R_j^{r_i} = \{t \in R_j | \pi_{r}(t) = r_i \}$ be the set of tuples in relation $R_j$ whose projection on $r$ is $r_i$. Let $T_i = \cup_{j \in [m]} R_j^{r_i}$. We separate the values in $\adom(r)$ into two different sets: heavy and light. Let $H = \{r_i \in \adom(r)| \frac{N}{\alpha} \leq |T_i|\}$ be the set of heavy $r_i$'s and let $L = \{r_i \in \adom(r)| \frac{N}{\alpha} > |T_i|\}$ be the set of light $r_i$'s.
We construct all the results having light value on the $r$ attribute i.e., $Q_L = \{t \in \Q(\I) | \pi_{r}(t) \in L\}=\bigcup\limits_{r_i \in L}(\Join_{j \in [m]} R_j^{r_i})$.
Then we construct an index $D^L$ for range counting (Theorem~\ref{rect}) over the tuples in $Q_L$. For an attribute $A\in\allattr$, let $\allrel_{A} = \{R_j\in\allrel | A \in \allattr_j \}$ be the set of relations having the attribute $A$. For every $r_i \in H$, we compute $R_j^{r_i}$, and then for each child attribute $c \in C$, we compute all the tuples $Q_{i}^{c} = \Join_{R_j \in \allrel_{c}}R_j^{r_i}$. For each value $r_i \in H$, and each child attribute $c \in C$, we construct a range counting index $D^H_{i,c}$ over the tuples in $Q_{i}^{c}$, using Theorem~\ref{rect}


\paragraph{Query processing}
  We are given a query rectangle $\rect$. For each $c\in C$, let $\allattr_c=\bigcup_{R_j\in\allrel_c}\allattr_j$ be the set of attributes contained in at least a relation in $\allrel_c$.
  Let
  $\pi_{\allattr_c}(\rect)$ be the rectangle defined by the linear predicates to the attributes in $\allattr_c$.
  Using the range tree $D^L$, we compute $n_1=D^L(\rect)$, which is the number of join results in the query rectangle with light values on $r$.
Next, compute the number of join results in the query rectangle with heavy values on $r$. For each $r_i \in H$ and for each $c\in C$, we use the range tree $D^H_{i,c}$ and we compute $h_{i,c}=D^H_{i,c}(\pi_{\allattr_c}(\rect))$, i.e., the number of tuples in $Q_i^c$ in rectangle $\pi_{\allattr_c}(\rect)$.
Let $n_2 = \sum_{r_i \in H}(\prod_{c \in C}h_{i,c})$.
We return $n_1+n_2$.



\fullversion{
\paragraph{Correctness} The correctness proof is similar to the correctness of the index for generalized $k$-star queries in Subsection~\ref{subsec:Gkstar}.
}

    \paragraph{Time and Space analysis}
    By the definition of the light values, we know that for each $r_i \in L$, we have $|R_j^{r_i}|=O(\frac{N}{\alpha})$. Let $t \in \I$, be an arbitrary tuple from an arbitrary relation $R_j$, and let $c_t = |\{t' \in Q_L | \pi_{\allattr_j}(t') = t\}|$. By the definition of $Q_L$, if $\pi_{r}(t) \in H$, then $c_t = 0$, and if $\pi_{r}(t) \in L$, then $c_t = |\Join_{j' \in [m], j' \neq j} R_j^{\pi_{r}(t)})|$. Therefore, by the bounds on the size of the results of joins, for
    any $t \in D$, we have $c_t \leq \left(\frac{N}{\alpha}\right)^{m - 1}$ and $|Q_L|\leq \frac{N^m}{\alpha^{m-1}}$, as we had
    in Section~\ref{subsec:Gkstar} for generalized $k$-star queries.
   \fullversion{ 
    $$c_t \leq \Pi_{j' \in [m], j' \neq j} |R_j^{\pi_{r}(t)})| \leq (\frac{N}{\alpha})^{m - 1}.$$
    Thus, we have
    $$
        |Q_L| \leq \sum_{j \in [m]}\sum_{t \in R_j}c_t
        \leq \sum_{j \in [m]}\sum_{t \in R_j}(\frac{N}{\alpha})^{m - 1}
        \leq N (\frac{N}{\alpha})^{m - 1}
        = \frac{N^m}{\alpha^{m-1}}.
    $$}
Using the generic join algorithm~\cite{ngo2014skew}, $Q_L$ can be computed in $\O(\frac{N^m}{\alpha^{m-1}})$ time.
     The index $D^L$ uses $\O(|Q_L|)=\O(\frac{N^m}{\alpha^{m-1}})$ space and can
     be constructed in $\O(|Q_L|)=\O(\frac{N^m}{\alpha^{m-1}})$ time.
    We execute one range count query on $D^L$ that takes $\O(1)$ time.
    By the definition of the heavy values, it is straightforward to see that $|H|=O(\frac{N}{N/\alpha}) = O(\alpha)$.
We run a range count query for each $r_i \in H$ on $D^H_{i,c}$ for every $c \in C$. In total, we run $O(|C|\cdot |H|)=O(\alpha)$ queries. Each query takes $\O(1)$ time, so the query procedure runs in $\O(\alpha)$ time.
It remains to analyze the space that is used for the heavy
values. Let $m_1 = \max_{c \in C}|\allrel_c|$.
By definition, it holds that $\max_{r_i\in H, c\in C}|Q_i^c|=O(\max_{j\in[m]} |R_j|^{m_1})$.
Every index $D_{i,c}^H$ uses near-linear space, so the total space used for all indexes $D_{i,c}^H$ is $\O(N^{m_1})$. All indexes $D_{i,c}^H$ are also constructed in $\O(N^{m_1})$ time.
Overall, our index answers a query in $O(\alpha)$ time using $\O(\frac{N^m}{\alpha^{m-1}} + N^{m_1})$ space having $\O(\frac{N^m}{\alpha^{m-1}} + N^{m_1})$ preprocessing time.

Recall that $m_1$ is the maximum number of relations among the subtrees rooted at children
$C$. The space highly depends on the structure of $\Q$'s attribute tree. If its attribute tree is
balanced and has fanout $d$, then the space is $\O(N^{m}/T^{m - 1}+N^{m/d})$. However, if the tree is highly unbalanced, then the space can be as large as $\O(N^{m}/T^{m - 1}+N^{m-1})$. For example, see Figure~\ref{fig:hier}.
Both queries have $6$
\begin{wrapfigure}{r}{0.4\textwidth}
    \centering
    \includegraphics[scale=0.3]{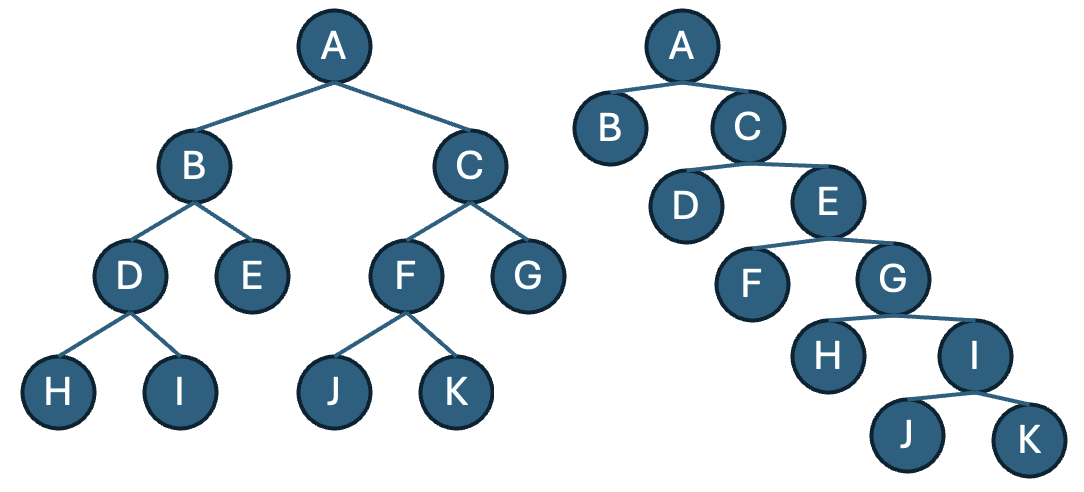}
    \vspace{-1em}
    \caption{Left: Attribute tree of a balanced hierarchical query with $6$ relations. Right: Attribute tree of an unbalanced hierarchical query with $6$ relations.}
    \label{fig:hier}
    \vspace{-1em}
\end{wrapfigure}
relations. If we set $T=\sqrt{N}$, then the space needed for the left (balanced) query is $O(N^{3.5})$ (since $m_1=3)$, while the space needed for the right (unbalanced) query is $O(N^5)$ (since $m_1=5$).

In Appendix~\ref{appndx:furtherOpt}, we show further optimizations that improve the space complexity of our index. The main idea is to recursively apply our technique by computing the light and heavy values in $\adom(c)$ for every root's child $c$. 
Using quadtrees (Theorem~\ref{nn}) instead of range trees, we can extend our index for the $\nnprob$ over hierarchical queries (similar to Appendix~\ref{appndx:annExt}).
We conclude with the next theorem that holds for both the $\prob$ and $\nnprob$ problems.


\begin{theorem}
\label{thm:mainHier2}

    Given a hierarchical CQ $\Q$ on a database $\I$ of size $N$ with $m$ relations, and a parameter $1\leq T\leq N$, there exists an index of 
    $\O\left(\min_{\ell\in \{0,\ldots \mathcal{L}-1\}} \frac{N^m}{T^{(m-1)/(\ell+1)}}+N^{m_{\ell+1}}\right)$ space and the same preprocessing time, where $\mathcal{L}$ is the depth of the attribute tree of $\Q$,
    and $m_{\ell+1}$ 
    is the maximum number of leaf nodes in the subtree rooted at a node of level $\ell+1$, such that given a query hyper-rectangle $\rect\in \Re^{|\head(\Q)|}$ (resp. point $q\in\Re^{|\head(\Q)|}$), it returns $\prob(\Q,\I,\rect)$ (resp. $\nnprob(\Q,\I,q)$) in $T$ query time. 
\end{theorem}
The index in Theorem~\ref{thm:mainHier2}, improves the space complexity over the right (unbalanced) query from Figure~\ref{fig:hier}. Using the index from Theorem~\ref{thm:mainHier2} the space complexity improves to $O(N^{4.75})$.

\subsection{Extensions to General CQs}
\label{subsec:extension}
We extend our indexes from hierarchical queries to arbitrary CQs using the hierarchical generalized hypertree decomposition (GHD) method~\cite{hu2022computing}.
Intuitively, the non-hierarchical query $\Q$ is transformed into a hierarchical one at the cost of increasing the input size from $N$ to $N^{\hhtw(\Q)}$,
where $\hhtw(\Q)$ is defined as the hierarchical hypertree width of the input query.
Since most of the extensions in this subsection are standard, we show the details in Appendix~\ref{appndx:extension} and we only state the final theorem.


\begin{theorem}
\label{thm:final}

    Given a hierarchical CQ $\Q$ on a database $\I$ of size $N$ with $m$ relations, and a parameter $1\leq T\leq N$, there exists an index of $\displaystyle{\O\left(\min_{\ell\in \{0,\ldots \mathcal{L}-1\}} \frac{N^{\hhtw(\Q)\cdot m}}{T^{(m-1)/(\ell+1)}}+N^{\hhtw(\Q)\cdot m_{\ell+1}}\right)}$ space and the same preprocessing time, where $\mathcal{L}$ is the depth of the attribute tree of the hierarchical decomposition of $\Q$,
    and $m_{\ell+1}$ 
    is the maximum number of leaf nodes in the subtree rooted at a node of level $\ell+1$, such that given a query hyper-rectangle $\rect\in \Re^{|\head(\Q)|}$ (resp. point $q\in\Re^{|\head(\Q)|}$), it returns $\prob(\Q,\I,\rect)$ (resp. $\nnprob(\Q,\I,q)$) in $T$ query time. 
\end{theorem}



\section{Improved Relational Algorithms}
\label{sec:extClust}
We show how our new indexes are used to improve the running time of known algorithms in the relational setting. 
There is a recent line of work solving combinatorial problems on the results of a CQ without explicitly computing $\Q(\I)$; for example enumeration~\cite{deep2022ranked, deep2021ranked, tziavelis2020optimal}, quantiles~\cite{tziavelis2023efficient},
direct access~\cite{carmeli2023tractable}, diversity~\cite{agarwal2025computing, arenas2024towards, merkl2025diversity}, top-$k$~\cite{tziavelis2020optimal}, and clustering~\cite{agarwal2025computing, chen2022coresets, curtin2020rk, esmailpour2025improved, moseley2021relational}.
In the rest of this section for simplicity, we focus on acyclic queries $\Q$~\cite{beeri1983desirability, fagin1983degrees}.

One of the main techniques in the papers above is to use oracles (indexes in our case) to get useful statistics over the results of a CQ. For example, in relational clustering~\cite{esmailpour2025improved} ($\mathcal{K}$-median and $\mathcal{K}$-means clustering problems) the authors used a (counting) oracle to compute $|\rect\cap\Q(\I)|$ given a rectangle $\rect$, or use a (sampling) oracle to get a sample from $\rect\cap \Q(\I)$ uniformly at random. In all cases, they describe straightforward oracles/indexes with $O(N)$ query time. The running time all the proposed algorithms in~\cite{esmailpour2025improved} can be written as $O(\tau\cdot \mathcal{T} + f(\mathcal{K}^2))$, where $\tau$ is the number of queries to the oracles, $\mathcal{T}$ is the query time of the oracle, and $f(\mathcal{K}^2)$ is the runtime of a clustering algorithm in the standard computational setting (i.e., input data stored in one table) over an input of size $\mathcal{K}^2$. More specifically, the authors describe a deterministic $O(\mathcal{K}^{2d+2}\cdot N+f(\mathcal{K}^2))$ time algorithm performing $\tau=O(\mathcal{K}^{2d+2})$ queries to a counting oracle and a randomized $O(\mathcal{K}^{2}\cdot N+f(\mathcal{K}^2))$ time algorithm performing $\tau=O(\mathcal{K}^2)$ queries to a counting and sampling oracle, where $d=|\head(\Q)|$.
For simplicity we skip the term $f(\mathcal{K}^2)$ from the running time because in most cases, using the fastest algorithms for clustering in the standard computational setting~\cite{har2003coresets}, $f(\mathcal{K}^2)\ll \mathcal{K}^2\cdot N$.
Both algorithms return a $(2+\eps)$-approximation solution for the relational $\mathcal{K}$-median problem and a $(4+\eps)$-approximation solution for the relational $\mathcal{K}$-means problem, where $\eps$ is an arbitrary small constant. We note that these are the fastest known $(2+\eps)$-approximation algorithms (resp. $(4+\eps)$-approximation algorithms) for the relational $\mathcal{K}$-median (resp. $\mathcal{K}$-means) problem.
In the next paragraphs we represent the query and preprocessing time of our indexes over a database instance $\I$ of size $N$ as $T(N)$ and $P(N)$, respectively.

We first consider the deterministic algorithm from~\cite{esmailpour2025improved}, where $\tau=O(\mathcal{K}^{2d+2})$ counting oracle queries are executed to compute $|\rect\cap \Q(\I)|$, given a rectangle $\rect$. Consider the query $\star$. We can use the index from the $\prob$ problem on $k$-star queries to compute $|\rect\cap \Q(\I)|$ when is needed. 
Using our result from Theorem~\ref{alg:star}, the deterministic algorithm runs in $O(P(N)+\tau\cdot T(N))=O(N+\frac{N^{k+1}}{T^{k+1}}+\mathcal{K}^{2d+2}\cdot T)$.
Our new algorithm is faster than the running time proposed in~\cite{esmailpour2025improved}, if $N+\frac{N^{k+1}}{T^{k+1}}+\mathcal{K}^{2d+2}\cdot T \ll \mathcal{K}^{2d+2}\cdot N$. This is indeed the case for larger values of $\mathcal{K}$. For example, assume $\mathcal{K}=N^{1/(2d+2)}$. Then the deterministic algorithm from~\cite{esmailpour2025improved} runs in $O(N^2)$ time, while our implementation runs in $O(\frac{N^{k+1}}{T^{k+1}}+N\cdot T)$. By selecting $T=N^{k/(k+2)}$, the running time of our implementation is $O(N^{(2k+2)/(k+2)})$ which is always asymptotically smaller than $O(N^2)$. For example, for $k=2$ (matrix query) our new algorithm runs in $O(N^{3/2})$ which is faster than the algorithm in~\cite{esmailpour2025improved}, by a $N^{1/2}$ factor.
We can make similar arguments for the $\path$ query using our index in Theorem~\ref{alg:path}, and general queries using our indexes from Section~\ref{sec:general}.

Next, we consider the randomized algorithm from~\cite{esmailpour2025improved}. It requires a counting oracle and a sampling oracle. Interestingly, by making slight modifications, all our indexes from Section~\ref{sec:general} can be extended to return uniform samples in $\rect\cap \Q(\I)$ with the same guarantees as in the $\prob$ problem. We call it the Range Sampling Query ($\sprob$) problem. We show the details in Appendix~\ref{appndx:RSQ}.
The randomized algorithm from~\cite{esmailpour2025improved} runs in $\O(\mathcal{K}^2\cdot N)$ time, performing $\tau=O(\mathcal{K}^2)$ queries to the counting and sampling oracles. Consider any generalized $k$-star query $\star$. Using our index from Theorem~\ref{range}, the randomized algorithm can be implemented in $O(N+\frac{N^k}{T^{k-1}}+\mathcal{K}^2\cdot T)$ time. The new algorithm is faster than the proposed algorithm in~\cite{esmailpour2025improved}, if $N+\frac{N^k}{T^{k-1}}+\mathcal{K}^2\cdot T\ll \mathcal{K}^2\cdot N$. If $\mathcal{K}=N^{1/2}$ then the randomized algorithm from~\cite{esmailpour2025improved} runs in $O(N^2)$ time, while our implementation runs in $O(\frac{N^k}{T^{k-1}}+N\cdot T)$. By selecting, $T=N^{(k-1)/k}$, the running time of our algorithm is $O(N^{(2k-1)/k})$ which is asymptotically smaller than $O(N^2)$ for every $k>1$.
We can make similar arguments for hierarchical and general queries using the results from Theorems~\ref{thm:mainHier2}, \ref{thm:final}.
We conclude with the next theorem.
\begin{theorem}
\label{thm:Clustimprovements}
  Given an acyclic CQ $\Q$, a database instance $\I$ of size $N$, a parameter $\mathcal{K}\in[1,N]$, and a constant parameter $\eps\in(0,1)$, there exists a deterministic (resp. randomized) $(2+\eps)$-approximation algorithm for the relational $\mathcal{K}$-median problem that runs in $O(\min\{\mathcal{K}^{2d+2}\cdot N, P(N)\!+\!\mathcal{K}^{2d+2}\!\cdot \!T(N)\}\!+\!f(\mathcal{K}^2))$ time (resp. $O(\min\{\mathcal{K}^{2}\!\cdot\! N, P(N)\!+\!\mathcal{K}^{2}\cdot T(N)\}\!+\!f(\mathcal{K}^2))$ time), where $d=|\head(\Q)|$, $f(\mathcal{K}^2)$ is the time complexity of a $\mathcal{K}$-median clustering algorithm in the standard computational setting, $P(N)$ is the preprocessing, and $T(N)$ is the query time of the new proposed indexes for the $\prob$ and $\sprob$ problems. 
  The algorithms can be extended to $(4+\eps)$-approximation algorithms for the relational $\mathcal{K}$-means problem with the same time guarantees.
\end{theorem}

\section{Conclusion}
In this paper, we design indexes with space-time tradeoffs for spatial conjunctive queries.
There are multiple interesting open problems derived from this work.
While we showed near-optimal index for the $\prob$ problem on $k$-star and $k$-path queries, it is interesting to study whether there exists a near-optimal index for the $\nnprob$ problem on $k$-star and $k$-path queries. Furthermore, an interesting problem is to close the gap between the space-time lower and upper bound for $\prob$ and $\nnprob$ on hierarchical and general CQs.



\bibliography{ref}
\newpage
\appendix

\section{Missing proofs from Section~\ref{sec:lb}}
\label{appndx:LB}

\begin{proof}[Proof of Theorem~\ref{lem:lowerbound2}] We consider a simpler version of $\path$ where $\mathcal{A}_1=A_1$, $\mathcal{A}_2=A_2$, $\mathcal{B}_i=B_i$ for every $i\in[k-1]$.
The reductions from $k$-reachability problem to $\prob$ or $\nnprob$ are similar; hence we focus on $\prob$ below.
Given an instance of the $k$-reachability problem over a graph $G=(V,E)$ with $N=|E|$ and a pair of vertices $v_{j_1}, v_{j_k}\in V$, we build a database instance $D$ such that using the answer of a \prob~ query over $\path(D)$, one can answer the $k$-reachability problem in $O(1)$ time.
Let $G'(V',E')$ be a graph such that $V'=V$, $E'=E\cup\{(v,v)\mid v\in V\}$, i.e., $G'$ is a copy of $G$ including self-loops edges $\{v,v\}$ for every node $v$ in $G'$. Notice that for any path from $v_{j_1}$ to $v_{j_k}$ in $G$ of length at most $k$, there exists a path from $v_{j_1}$ to $v_{j_k}$ in $G'$ of length exactly $k$, and vice versa. Hence, conjecture~\ref{conj:reach} also holds for graph $G'$ for paths of length exactly $k$. 

Initially, $D=\emptyset$.
For every directed edge $(v_{h_1},v_{h_2})\in E'$ we add the tuple $(h_1,h_2)$ in $R_i^D$ for every $i\in [k]$. Consider the $k$-reachability query $v_{j_1}, v_{j_k}\in V$ for a path of length exactly $k$. Let $\rect=[j_1,j_1]\times [j_k,j_k]$ be a rectangle in $\Re^2$.

  If the $k$-reachability instance is true then there is a directed path $v_{j_1}\rightarrow v_{j_2}\rightarrow\ldots \rightarrow v_{j_k}$  in $G'$ of length exactly $k$. Hence, for every $i\in[k]$, $(v_{j_i},v_{j_{i+1}})\in E'$ and the tuple $(v_{j_i},v_{j_{i+1}})$ exists in $R_i^D$. Since $\rect=[j_1,j_1]\times [j_k,j_k]$ and $|\Join_{i\in[k]}(v_{j_i},v_{j_{i+1}})|>0$, the query $\prob(\path,D,\rect)$ returns a counter greater than $0$.

  If $\prob(\path, D,\rect)>0$ then for every $i\in[k]$, there exists a tuple $(v_{j_i}, v_{j_{i+1}})$ in $R_i^D$. Hence, every edge $(v_{j_i}, v_{j_{i+1}})\in E'$ for $i\in [k]$, so there exists a path $v_{j_1}\rightarrow v_{j_2}\rightarrow\ldots \rightarrow v_{j_k}$ in $G'$ of length exactly $k$. 
\end{proof}

\section{Missing proofs from Section~\ref{sec:opt}}
\label{appndx:sec4}
 \begin{proof}[Proof of Lemma~\ref{lem:domcnt}]
    We construct the index as follows. First, the dangling tuples (tuples that do not participate in any join) are removed~\cite{beeri1983desirability}. We fix an arbitrary relation $R_j$. For every $t_{b'}\in \pi_{\mathcal{B}}(R_j)$, and for every $i\in[k]$ we construct a range tree $\rangetree_{t_{b'}}^{(i)}$ for range count queries over the points in $P_{t_{b'}}^{(i)}=\{t\in \proj_{\mathcal{A}_i}(R_i)\mid \exists t_i\in R_i, \pi_{\mathcal{B}}(t_i)=t_{b'}, \pi_{\mathcal{A}_i}(t_i)=t\}$. Each tree $\rangetree_{t_{b'}}^{(i)}$ uses $\O(|P_{t_{b'}}^{(i)}|)$ space. We note that every point is stored in one tree, and the total size of $\I$ is $N$, so the index uses $\O(N)$ space. Furthermore, by the definition of range trees the index can be constructed in $\O(N)$ time. For a query rectangle $\rect$, let $\rect_i=\pi_{\mathcal{A}_i}(\rect)$ be the hyper-rectangle considering only the linear predicates of the attributes/coordinates in $\mathcal{A}_i$ for $i\in[k]$. Given a query tuple $b$, for every $i\in[k]$, and a query rectangle $\rect$, we compute $n_i=\rangetree_{b}^{(i)}(\rect_i)$, i.e., the number of tuples in $R_i$ whose projection on $\mathcal{B}$ is $b$ and whose projection on $\mathcal{A}_i$ lies in $\rect_i$. In the end we return $\prod_{i\in[k]}n_i$. The query procedure spends $\O(1)$ time for every range tree query, so in total, it runs in $\O(k)=\O(1)$ time.
    \end{proof}

\begin{proof}[Proof of Lemma~\ref{lem:spaceKstar}]
     The index $\primitive$ uses $\O(N)$ space.
     Every tree $\rangetree_i$ uses $\O(\frac{N}{\alpha})$ space. Finally, we store one number for every combination of nodes $u_1,\ldots, u_k$. $O(N/\alpha)$ nodes/blocks are stored in every tree $\rangetree_i$, so $O(N^k/\alpha^k)$ different combinations exist. Overall, the index uses $\O(N+N^k/\alpha^k)$ space.

     The index $\primitive$ is constructed in $\O(N)$ time. Every tree $\rangetree_i$ is constructed in $\O(N/\alpha)$ time. For every combination of nodes $u_1,\ldots, u_k$, the value $n_{u_1,\ldots, u_k}$ is computed in $\O(N/\alpha)$ time by a straightforward range counting query on the results of a join query with linear query time, as shown in~\cite{agarwal2025computing}.
 \end{proof}

\begin{proof}[Correctness proof of query procedure in Section~\ref{subsec:kstar}]
We show that every tuple in the multi-set $\star(\I)\cap \rect$ is counted exactly one.
Let $t\in \star(\I)\cap \rect$ be a tuple in the multi-set. Notice that the there might be multiple copies of $t$ in $\star(\I)\cap \rect$. We consider that $t$ is the copy that was created by exactly joining the tuples $t_j\in R_j$, for $j\in[k]$.
Assume that $\pi_{A_i}(t_i)\in \mathcal{I}_i$ for every $i\in[k]$. By the definition of the range-trees $\mathcal{T}_i$, there exists a unique combination $u_{(1,j_1)}\in U_1,\ldots u_{(k,j_k)}\in U_k$ such that $\pi_{A_i}(t_i)\in I_i^{(u_{(i,j_i)})}$ for $i\in[k]$. Hence $t$ is counted once by $n_{u_{(1,j_1)},\ldots, u_{(k,j_k)}}$ and equivalently $t$ is counted once by $\ell_1$.
Then, we assume that $t\in \star(\I)\cap \rect$ satisfies $\pi_{A_i}(t_i)\in\mathcal{I}_i$ for $i<h$, and $h$ is the first attribute such that $\pi_{A_i}(t_i)\in I_{1,j_1}$ (the case where $\proj_{A_i}(t_i)\in I_{1,j_1'}$ is equivalent). By definition $t\in\rect_{t_h}$ and if $b=\pi_B(t_h)$ then $t$ is only counted once by $\primitive.\mathsf{count}(b,\rect_{t_h})$. So it is counted once by $\ell_2^{(h)}$. Finally, by definition, it is straightforward that we count only results that lie in the rectangle $\rect$.

\end{proof}

\paragraph{Extension to $k$-star queries with multiple attributes}
For $k$-star queries, as defined in Section~\ref{subsec:conj} ($\mathcal{A}_i$ and $\mathcal{B}$ might contain more than one attribute), only the construction of buckets are changing. We first  construct a slightly modified $|\mathcal{A}_i|$-dimensional range tree $\rangetree_i$ on $\proj_{\mathcal{A}_i}(R_i)$, for every $i\in[k]$. In particular, we construct $\rangetree_i$ as a standard $|\mathcal{A}_i|$-dimensional range tree in the first $|\mathcal{A}_i|-1$ levels. In the last level, we have the nodes of the subtree of the standard $|\mathcal{A}_i|$-dimensional range tree that contain at least $\alpha$ points. In other words, we remove the nodes of the standard $|\mathcal{A}_i|$-dimensional range tree that contain less than $\alpha$ points.
Every node from the last level of $\rangetree_i$ defines a block $\block_{i,j}$ of tuples in $\R_i$.
Now, each $I_i$ is not a set of intervals but a set of hyper-rectangles defined by the nodes in the last level of the range trees.
As we had in the simpler case, for any combination of nodes (in the last level of their range tree), $u_1\in \rangetree_1, \ldots, u_k\in\rangetree_k$, we compute $n_{u_1,\ldots, u_k}$. Given a query rectangle $\rect$, let $\rect_i=\pi_{\mathcal{A}_i}(\rect)$ for every $i\in[k]$. Let $U_i=\{u_{(i,1)},\ldots, u_{(i,\mu_i)}\}$ be the set of canonical nodes in $\rangetree_i$ after querying with $\rect_i$, and let $\mathcal{I}_i'$ be their corresponding set of canonical rectangles. By the query procedure of the range trees, there are $\O(1)$ canonical nodes and $\O(1)$ partially intersected nodes in every $\rangetree_i$.
The rest of the query procedure is almost identical to the simpler case $|\mathcal{A}_i|=1$. Instead of using the interval $\mathcal{I}_i$, we use every rectangle in $\mathcal{I}_i'$ (for every $i\in[k]$) to compute $\ell_2^{(i)}$. The overall space and query time remain the same (up to $\polylog(n)=\O(1)$ factors).

 \begin{proof}[Proof of Lemma~\ref{lem:spaceKpath}]
     Let $\mathcal{S}_k$ denote the space of the index we construct for the $\prob$ problem on $k$-path queries.
     Every tree $\rangetree_1, \rangetree_2$ uses $O(\frac{N}{\alpha})$ space. We store one number for every combination of nodes $u_1, u_2$. There are $O(N/\alpha)$ blocks stored in $\rangetree_1$ and $\rangetree_2$, so there exist $O(N^2/\alpha^2)$ different combinations. Hence, it holds that $\mathcal{S}_k=2\mathcal{S}_{k-1}+N^2/\alpha^2$. By theorem~\ref{alg:star}, $\mathcal{S}_2=\O(N+N^2/\alpha^2)$, and using the fact that $k$ is constant, we conclude $\mathcal{S}_k=\O(N+N^2/\alpha^2)$.

     Let $\mathcal{P}_k$ denote the preprocessing time of the index we construct for the $\prob$ problem on $k$-path queries.
     Every tree $\rangetree_1, \rangetree_2$ is constructed in $\O(\frac{N}{\alpha})$ time.
     For every combination of nodes $u_1,u_2$ the value $n_{u_1,u_2}$ is computed in $\O(N/\alpha)$ time using the counting oracle in~\cite{agarwal2025computing}.
    There are $O(N/\alpha)$ blocks stored in $\rangetree_1$ and $\rangetree_2$, so there exist $O(N^2/\alpha^2)$ different combinations. Hence, it holds that $\mathcal{P}_k=2\mathcal{P}_{k-1}+N^3/\alpha^3$. By theorem~\ref{alg:star}, $\mathcal{S}_2=\O(N+N^3/\alpha^3)$, and using the fact that $k$ is constant, we conclude $\mathcal{P}_k=\O(N+N^3/\alpha^3)$.
 \end{proof}

\begin{proof}[Correctness proof of query procedure in Section~\ref{subsec:path}]
We show that every tuple in the multi-set $\path(\I)\cap \rect$ is counted exactly one.
Let $t\in \path(\I)\cap \rect$ be a tuple in the multi-set. Notice that the there might be multiple copies of $t$ in $\path(\I)\cap \rect$. We consider that $t$ is the copy that was created by exactly joining the tuples $t_j\in R_j$, for $j\in[k]$.
By induction hypothesis, assume that the indexes $\ds_{k-1}^{(1)}$ and $\ds_{k-1}^{(2)}$ are correct. 
Let $t\in \star(\I)\cap \rect$ such that $\pi_{A_1}(t_1)\in \mathcal{I}_1$ and $\pi_{A_2}(t_k)\in \mathcal{I}_2$. Then, $t$ is counted once by $\ell_1$, as we had in $k$-star queries.
Then, we assume that $t\in \star(\I)\cap \rect$ satisfies $\pi_{A_1}(t_1)\notin\mathcal{I}_1$ (the proof is equivalent if $\pi_{A_1}(t_1)\in\mathcal{I}_1$ and $\pi_{A_2}(t_k)\notin\mathcal{I}_2$).
Let $b=\pi_B(t_1)$. The quantity 
$\ds_{k-1}^{(2)}(\rect_{t_h})$ counts all results in $\Gpath_{k-1}(B_1,A_{2})$ for $B_1=b$ and $A_2\in\rect_2$. The correctness follows by the correctness of $\ds_{k-1}^{(2)}$.

\end{proof}

\subsection{Additional optimal queries}
\label{appndx:additional}
A trivial lower bound for the $\prob$ and $\nnprob$ problems is an index with $O(n)$ space and $\O(1)$ query time. In this subsection, we show a family of conjunctive queries where such an optimal index can be constructed.

\begin{theorem}
    Given a CQ $\Q$ on a database $\I$ of size $N$, if there exists $j\in[m]$ such that $\head(\Q)\subseteq \allattr_j$, then there exists an index with $\O(N)$ space and $\O(1)$ query time for the $\prob$ and the  $\nnprob$ problem.
\end{theorem}
\begin{proof}
    For any tuple $t\in\proj_{\head(\Q)}(R_j)$ we compute $w_t=|\{t'\in\Q(\I)\mid \pi_{\head(Q)}(t')=t\}|$. We construct a $|\head(\Q)|$-dimensional range tree for the weighted version of range counting queries, over $\proj_{\head(\Q)}(R_j)$. The index has $\O(|R_j|)=\O(N)$ space. Given a query rectangle $\rect$ we run a (weighted) range counting query in $\O(1)$ time. Similarly, for the $\nnprob$, after keeping only the non-dangling tuples in the database we construct an approximate nearest neighbor index (quadtree) over the points in $\proj_{\head(\Q)}(R_j)$. The result follows.
\end{proof}

\section{Missing proofs from Section~\ref{sec:general}}
\label{appndx:sec5}

\begin{proof}[Proof of Lemma~\ref{lem:equiv}]
Given $\rect\in\Re^{|\mathbf{y}|}$ we define a rectangle $\rect'\in\Re^{|\bigcup_{i\in[m]}\mathbf{A}_i|}$ as follows. For every $A\in \mathbf{y}$, we set $\pi_A(\rect')=\pi_A(\rect)$. For every $A\notin \mathbf{y}$, we set $\pi_A(\rect')=(-\infty,\infty)$. Then it is straightforward to see that $|\Q(\I)\cap \rect|=|\Q'(\I)\cap \rect'|$, under bag semantics. Hence an index for the $\prob$ problem on the full CQ $\Q'$ can be used to answer $\prob(\Q,\I,\rect)$.

\end{proof}

\begin{proof}[Correctness proof of query procedure in Section~\ref{subsec:Gkstar}]
$D^L(\rect)$ counts every tuple \\$t\in\star(\mathcal{A}_1,\ldots,\mathcal{A}_k,\mathcal{B})\cap \rect$ with a light projection on $\mathcal{B}$, i.e., $\pi_{\mathcal{B}}(t)\in L$. Similarly, by definition, $\sum_{b_i \in H}(\prod_{j \in [k]}h_{i,j})$ counts all tuples $t\in\star(\mathcal{A}_1,\ldots,\mathcal{A}_k,\mathcal{B})\cap \rect$ such that $\pi_{\mathcal{B}}(t)=b_i\in H$. The result follows.
\end{proof}

\begin{proof}[Time and Space analysis of our index in Section~\ref{subsec:Gkstar}]
    By the definition of the light values, we know that for each $b_i \in L$, we have $|R_j^{b_i}|=O(\frac{N}{\alpha})$. Let $t \in R_j^\I$ be an arbitrary tuple. Let $c_t = |\{t' \in Q_L | \pi_{\mathcal{A}_j}(t') = \pi_{\mathcal{A}_j}(t)\}|$ be the number of join results with light value on $\mathcal{B}$ whose projection with respect to the attributes $\mathcal{A}_j$ is $\pi_{\mathcal{A}_j}(t)$. By the definition of $Q_L$, if $\pi_{\mathcal{B}}(t) \in H$, then $c_t = 0$, and if $\pi_{\mathcal{B}}(t) \in L$, then $c_t = |\Join_{j' \in [k], j' \neq j} R_j^{\pi_{\mathcal{B}}(t)}|$. Therefore, by the bounds on the size of the results of joins~\cite{atserias2013size}, for any $t \in R_j$, we have 
    $c_t \leq \Pi_{j' \in [k], j' \neq j} |R_j^{\pi_{\mathcal{B}}(t)}| \leq \left(\frac{N}{\alpha}\right)^{k - 1}.$
    Thus, we have
    $
        |Q_L| \leq \sum_{j \in [k]}\sum_{t \in R_j}c_t
        \leq \sum_{j \in [k]}\sum_{t \in R_j}\left(\frac{N}{\alpha}\right)^{k - 1}
        \leq N \left(\frac{N}{\alpha}\right)^{k - 1}
        = \frac{N^k}{\alpha^{k-1}}.
    $
   Hence, $D^L$ uses $\O(Q_L)=\O(\frac{N^k}{\alpha^{k-1}})$ space. Using the generic join algorithm~\cite{ngo2014skew}, $Q_L$ is also computed in $O(\frac{N^k}{\alpha^{k-1}})$ time.
   In the query procedure, we run one range count query on $D^L$ that takes $\O(1)$ time. By the definition of the heavy values, it is straightforward to see that $|H|=O(\frac{N}{N/\alpha}) = O(\alpha)$. 
    For each $b_i\in H$ and $R_j\in \allrel$ the index $D^H_{i,j}$ uses $\O(|R_j^{b_i}|)$ space. Notice that $\sum_{b_i \in H}\sum_{j \in [k]}R^{b_i}_{j} = O(N)$ so in total, all range trees $D^H_{i,j}$ use $\O(N)$ space and can be constructed in $\O(N)$ time.
    For each $b_i \in H$ and $R_j\in \allrel$ we run a count query using $D^H_{i,j}$ in $\O(1)$ time to compute $h_{i,j}$. We run $O(\alpha\cdot k)$ count queries, so the query time is $\O(\alpha)$.
\end{proof}

\subsection{Extension to $\nnprob$}
\label{appndx:annExt}
Our ideas can be extended to handle the $\nnprob$ problem.
Instead of using the range tree indexes from Theorem~\ref{rect}, we use the Approximate Nearest Neighbor index from Theorem~\ref{nn}. More specifically, we construct the ANN index $D^L$ on $Q_L$, and the ANN index $D^{H}_{i,j}$ on $\pi_{\mathcal{A}_j}(R_j^{b_i})$, for every $b_i\in H$ and $j\in [k]$.

Given a query point $q$, we first use the index $D^L$ to find a tuple $t^{(L)}\in Q_L$ such that $\dist(q,t^{(L)})\leq (1+\eps)\dist(q,Q_L)$. Next, we go through the heavy values one by one. For each $b_i\in H$, and for every $R_j\in \allrel$, we use $D^H_{i,j}$ to run the approximate nearest neighbor query on the query point $\pi_{\mathcal{A}_j}(q)$. More formally, we set 
$t^{(H)}_{i,j}= D^H_{i,j}(\pi_{\mathcal{A}_j}(q))$ for every $j\in[k]$. We define $t_i^{(H)}= (t_{i,1}^{(H)},t_{i,2}^{(H)},\ldots, t_{i,k}^{(H)}, b_i)$, and we set $t^{(H)}=\argmin_{t_i^{(H)}}\dist(q,t_i^{(H)})$.
We return $t'=\argmin_{t\in \{t^{(L)}, t^{(H)}\}}\dist(q,t)$.

\begin{lemma}
\label{lem:eNN}
    It holds that $\dist(q,t')\leq (1+\eps)\dist(q,\star(\I))$.
\end{lemma}
\begin{proof}
    Let $t^*$ be the nearest neighbor in $\star$ from the point query $q$. If $\pi_{B}(t^*)\in L$, by definition $\dist(q,t')\leq \dist(q,t^{(L)})\leq (1+\eps)\dist(q,t^*)$. Then, we assume $\pi_{\mathcal{B}}(t^*)=b_i\in H$. By definition of $D_{i,j}^H$, we have $\dist(\pi_{\mathcal{A}_j}(q),t_{i,j}^{(H)})\leq(1+\eps)\dist(\pi_{\mathcal{A}_j}(q),\pi_{\mathcal{A}_j}(t^*))$.
    Hence, 
    \begin{align*}
    \dist(q,t^{(H)})\!\!\leq\!\! \dist(q,t_i^{(H)})\!=\!\!\sqrt{\sum_{j\in[k]}\!\!\!\dist^2\left(\pi_{\mathcal{A}_j}(q),t_{i,j}^{(H)}\right)}\!\!\leq\!\! \sqrt{\sum_{j\in[k]}\!\!(1+\eps)^2\dist^2\left(\pi_{\mathcal{A}_j}(q),\pi_{A_j}(t^*)\right)}\!=\!(1+\eps)\dist(q,t^*).    
    \end{align*}
    \vspace{-1em}
\end{proof}

While the index is described for join queries, it can be extended to CQs by constructing the approximate nearest neighbor index $D^L$ on $\pi_{\head(\Q)}(\Q_L)$ and the index $D_{i,j}^H$ on $\pi_{\mathcal{A}_j\cap \head(\Q)}(R_j^{b_i})$.

\subsection{Further optimizations}
\label{appndx:furtherOpt}
In Section~\ref{subsec:hierarchical} we store $Q_L$, i.e., the tuples in $\Q(\I)$ whose projection on $r$ is light. Then, for any heavy value $r_i$ and any child $c$ we stored $Q_i^c$. We note that the set $Q_i^c$ is the result of a hierarchical query $\Join_{R_j\in R_c}R_j^{r_i}$. Hence, instead of pre-computing $Q_i^c$, the main idea is to recursively apply our technique by computing the light and heavy values in $\adom(c)$ for every root's child $c$. Equivalently, we can recursively apply this idea to the nodes in the next level, and so on. Let $\ell$ be the last level such that in all nodes on levels $0,1,\ldots, \ell$ we applied the heavy-light decomposition.
For a node $v$ in the attribute tree, let $\Q_v(\I)$ be the result of the query defined by the relations in $\allrel_v$. If $V_{h}$ is the set of nodes at level $h$, we define $m_{h}=\max_{u\in V_{h}}|\allrel_u|$.
In every node $v$ of the attribute tree up to level $\ell$, we compute $Q_L^{v}=\{t\in\Q_v(\I)\mid \pi_{v}(t)\in L_v\}$, where $L_v=\{v_i\in \adom(v)\mid \frac{N}{\alpha}>|T_i|\}$, and $T_i=\bigcup_{R_j\in\allrel_{v}}R_j^{v_i}$.
We construct a range counting index $D_v^L$ on $Q_L^{v}$. Since the attribute tree has constant number of nodes, we have that all
indexes $D_v^L$ use $\O(\frac{N^m}{\alpha^{m-1}}+\sum_{h\in [\ell]}\frac{N^{m_h}}{\alpha^{m_h-1}})=\O(\frac{N^m}{\alpha^{m-1}})$ space.
Every heavy value $v_i\in H_v$ in a node $v$ up to level $\ell$, where $H_v=\{v_i\in \adom(v)\mid \frac{N}{\alpha}\leq |T_i|\}$, has pointers to the children of $v$ in the attribute tree.
For every value $v_i\in \adom(v)$ and $v$'s child node $u$ at level $\ell+1$, we construct a range counting index $D_{i,u}^H$ over $\Q_i^u$. Hence, the total space used for all indexes $D_{i,u}^H$ is $\O(N^{m_{\ell+1}})$. Overall, the new index uses $\O(\frac{N^m}{\alpha^{m-1}}+N^{m_{\ell+1}})$ space.

Given a query rectangle $\rect$, we run $\O(\alpha^{\ell})$ queries to indexes $D_{v}^L$ and $\O(\alpha^{\ell+1})$ queries to indexes $D_{i,u}^H$. In total, the query time is $\O(\alpha^{\ell+1})$. If we set the query time to be $T$, then the space is $\O(\frac{N^m}{T^{(m-1)/(\ell+1)}}+N^{m_{\ell+1}})$. Since we do not know the level of the attribute tree that the space upper bound is minimized, we evaluate the asymptotic space complexity in every level and we select the minimum one (it can be done in $O(1)$ time in data complexity). Let $\mathcal{L}=O(1)$ be the maximum level (depth) of the attribute tree of the hierarchical query $\Q$. If the query time is $T$ then the space of our index is $\O\left(\min_{\ell\in \{0,\ldots \mathcal{L}-1\}} \frac{N^m}{T^{(m-1)/(\ell+1)}}+N^{m_{\ell+1}}\right)$. 

The proposed index also works for the \nnprob\ problem constructing quadtrees instead of range trees for the $D_{i,u}^H$ indexes, and using the ideas from the proof of Lemma~\ref{lem:eNN}.

\subsection{Missing details in Section~\ref{subsec:extension}}
\label{appndx:extension}
\begin{definition}[Generalized Hypertree Decomposition]
\label{def:ghd}
Given a join query $\Q=(\V,\E)$, a GHD of $\Q$ is a pair $(\T, \lambda)$, where
$\T$ is a tree as an ordered set of nodes and $\lambda: \T \to 2^{\V}$ is a labeling function which associates to each vertex $u \in \T$ a subset of attributes in $\V$, $\lambda_u$, such that the following conditions are satisfied:
i) (coverage) For each $e \in \E$, there is a node $u \in \T$ such that $e \subseteq \lambda_u$;
ii) (connectivity) For each $x \in \V$, the set of nodes $\{u \in \T: x \in \lambda_u\}$ forms a connected subtree of $\T$.
\end{definition}

\newcommand{\x}{\chi}
A fractional edge cover of a join query $\Q = (\V, \E)$ is a point $\x = \{\x_e \mid e\in \E\} \in \mathbb{R}^{|\E|}$ such that for any vertex $v \in \V$, $\sum_{e \in \E_v} \x_e \ge 1$.
As proved in~\cite{atserias2013size}, the maximum output size of a join query $\Q$ is $O(N^{\lVert \x \rVert_1})$. Since the above bound holds for any fractional edge cover, we define $\rho = \rho(\Q)$ to be the fractional cover with the smallest $\ell_1$-norm, i.e., $\rho(\Q)$ is the value of the objective function of the optimal solution of linear programming (LP):
\begin{equation}
\label{eq:lp}
    \min \sum_{e \in \E} \x_e, \;\text{s.t.}\; \forall e \in E: \x_e \ge 0 \;\text{and}\; \forall v \in \V: \sum_{e \in \E_v} \x_e \ge 1.
\end{equation}

Given a join query $\Q$, one of its GHD $(\T, \lambda)$ and a node $u \in \T$, the width of $u$ is defined as the optimal fractional edge covering number of its derived hypergraph $(\lambda_u, \E_u)$, where $\E_u = 
\{e\cap \lambda_u: e \in \E\}$. Given a join query and a GHD $(\T, \lambda)$, the width of $(\T, \lambda)$ is defined as the maximum width over all nodes in $\V_\T$. Then, the hierarchical hypertree width of a join query follows:

\begin{definition}[Hierarchical Hypertree Width] 
\label{def:hhtw}
The hierarchical hypertree width of a join query $\Q$, denoted as $\hhtw(\Q)$, is
\[\hhtw(\Q) = \min_{(\T, \lambda):\T \textrm{ is hierarchical}} \max_{u \in \T} \rho(\lambda_u, \E_u)\]
i.e., the minimum width over all hierarchical GHDs. 
\end{definition}

Basically, $O(N^{\hhtw(\Q)})$ is an upper bound on the number of join results materialized for each node in $\T$.
Using the definitions above, the algorithm for general CQs is the following. Given a CQ $\Q$ we construct the hierarchical GHD $(\T,\lambda)$ with the minimum width $\hhtw(\Q)$ over all hierarchical GHDs.
For each node $u$ of $\T$, we evaluate all query results among the relations that correspond to $\E_u$ and store them in a new relation $R_u'$. Then, we notice that $\Q':-\Join_{u\in \T}R_u'(\lambda_u)$ is a hierarchical query, so we construct the index from Theorem~\ref{thm:mainHier2} over $\Q'$.


\section{Indexes for the $\sprob$ problem}
\label{appndx:RSQ}
Given a CQ $\Q$ over a database instance $\I$ of size $N$ the goal is to construct an index, such that given a query rectangle $\rect\in \Re^{|\head|}$, a sample from the multi-set $\rect\cap \Q(\I)$ is returned uniformly at random. We show that all indexes from Section~\ref{sec:general} can be straightforwardly modified to to solve the $\sprob$ problem.

First, we claim that we can focus on full CQs as in Lemma~\ref{lem:equiv}. Indeed, if we have an index $D$ for the $\sprob$ problem on a full CQ $\Q'$ oer a set of attributes $\allattr$, then we can use the same index with the same guarantees for the $\sprob$ problem on the CQ $\Q$ with the same body attributes and relations as $\Q'$, where $\head(\Q)\subset \allattr$. Given the query hyper-rectangle $\rect\in\Re^{|\head(\Q)|}$, we construct the rectangle $\rect'\in \Re^{|\allattr|}$ as follows. For every $A\in \head(\Q)$, we set $\pi_A(\rect')=\pi_A(\rect)$. For every $A\notin \head(\Q)$, we set $\pi_A(\rect')=(-\infty,\infty)$. Using $D$, let $t$ be the sample returned from $\rect'\cap \Q'(\I)$ uniformly at random. We return $\pi_{\head{(\Q)}}(t)$. Every tuple in the set $\rect'\cap \Q'$ is mapped to exactly one tuple in the multi-set $\rect\cap \Q$ and vice versa. Hence, $\pi_{\head{(\Q)}}(t)$ is a sample from $\rect\cap \Q$ selected uniformly at random.

Next, we consider the full $\star$ query. The index is similar to the index from Section~\ref{subsec:Gkstar}. Let $D^*$ be the index from Section~\ref{subsec:Gkstar}. Next, we partition the values in $\adom(\mathcal{B})$ as light and heavy values. We follow the notation from Section~\ref{subsec:Gkstar}. All results $Q_L$ generated by light values are stored in a range tree $D^L$.
It is known that range trees can return a sample uniformly at random in a query hyper-rectangle in $\O(1)$ time~\cite{afshani2019independent, afshani2017independent}.
Next, for each heavy value $b_i$ and for each relation $R_j$ for $j\in[m]$ we construct a range tree $D_{i,j}^H$ (for sampling uniformly at random) over $R_{j}^{b_i}$. Following the analysis from Section~\ref{sec:general}, the index takes $\O(N+N^k/T^{k-1})$ space and can be constructed in $\O(N+N^k/T^{k-1})$ time. Given a query hyper-rectangle $\rect$, we use $D^*$ to count $w_1=|Q_L\cap \rect|$ and $w_2=|Q_H\cap \rect|$, where $Q_L$ (resp. $Q_H$) is the set of query results generated by light (resp. heavy) values. With probability $w_1/(w_1+w_2)$, we return a sample generated by light values, while with probability $w_2/(w_1+w_2)$ we return a sample generated by a heavy value. Assume that we return a join result generated by a light value. Using the range tree $D^L$, we return a sample from $Q_L\cap \rect$ uniformly at random. Next, assume that we return a join result generated by a heavy value. Using $D^*$, for every heavy value $b_i\in \mathcal{B}$, we get $z_i=|\{t\in \Q(\I)\mid \pi_{\mathcal{B}}(t)=b_i\}|$. We note that $\sum_{b_h\in H}z_h=w_2$, where $H$ is the set of heavy values. We sample a heavy value $b_i$ with probability $\frac{z_i}{\sum_{b_h\in H}z_h}$. If the value $b_i\in H$ is selected uniformly at random, then we use $D_{i,j}^H$ and return a random sample $s_j$ from $R_j^{b_i}$, for every $j\in[k]$. In the end, we return $\times_{j\in[k]}s_j$. We can easily observe that in all cases we return a valid sample from $\rect\cap \Q$ uniformly at random, using the correctness of $D^*$ and the correct counts $w_1,w_2, z_1,\ldots, z_{|H|}$. For example, a query result generated by a light value is selected with probability $\frac{w_1}{w_1+w_2}\cdot \frac{1}{w_1}=\frac{1}{w_1+w_2}$. A query result, generated by the heavy value $b_i$ is selected with probability $\frac{w_2}{w_1+w_2}\cdot \frac{z_i}{w_2}\cdot\frac{1}{z_i}=\frac{1}{w_1+w_2}$.
Finally, it follows straightforwardly from Section~\ref{subsec:Gkstar} that the query time of our index is $T$.

Similarly, we can have an index for the $\sprob$ problem over hierarchical queries as we had in Section~\ref{subsec:hierarchical}.
Let $D^*$ be the index from Section~\ref{subsec:hierarchical}. Next, we partition the values in $\adom(r)$ as light and heavy values. We follow the notation from Section~\ref{subsec:hierarchical}. All results $Q_L$ generated by light values are stored in a range tree $D^L$. Next, for each heavy value $r_i$ and for each child (in the attribute tree) $c$ we construct a range tree $D_{i,c}^H$ (for sampling uniformly at random) over $Q_i^c$. Following the analysis from Section~\ref{sec:general}, the index takes $\O(N^m/T^{m-1}+N^{m_1})$ space and can be constructed in $\O(N^m/T^{m-1}+N^{m_1})$ time. Given a query hyper-rectangle $\rect$, we use $D^*$ to count $w_1=|Q_L\cap \rect|$ and $w_2=|Q_H\cap \rect|$, where $Q_L$ (resp. $Q_H$) is the set of query results generated by light (resp. heavy) values. With probability $w_1/(w_1+w_2)$, we return a sample generated by light values, while with probability $w_2/(w_1+w_2)$ we return a sample generated by a heavy value. Assume that we return a join result generated by a light value. Using the range tree $D^L$, we return a sample from $Q_L\cap \rect$ uniformly at random. Next, assume that we return a join result generated by a heavy value. Using $D^*$, for every heavy value $r_i$, we get $z_i=|\{t\in \Q(\I)\mid \pi_c(t)=r_i\}|$. We sample a heavy value $r_i$ with probability $\frac{z_i}{\sum_{r_h\in H}z_h}$, where $H$ is the set of heavy values. If the value $r_i\in H$ is selected uniformly at random, then we use $D_{i,c}^H$ and return a random sample $s_c$ from $Q_i^c$, for every child $c$. In the end, we return $\times_{c\in C}s_c$. We can easily observe that in all cases we return a valid sample from $\rect\cap \Q$ uniformly at random, using the correctness of $D^*$ and the counts $w_1,w_2, z_1,\ldots, z_{|H|}$.
Finally, it follows straightforwardly from Section~\ref{subsec:hierarchical} that the query time of our index is $T$.

Equivalently, Theorems~\ref{thm:mainHier2},~\ref{thm:final} apply for the $\sprob$ problem, straightforwardly.

\begin{corollary}
    For every index for the $\prob$ problem in Section~\ref{sec:general}, there exists an index with the same space, preprocessing, and query time gurantees for the $\sprob$ problem.
\end{corollary}

\begin{figure*}[h!t]
	\centering
    \begin{minipage}[t]{0.32\textwidth}
		\includegraphics[width=\textwidth]{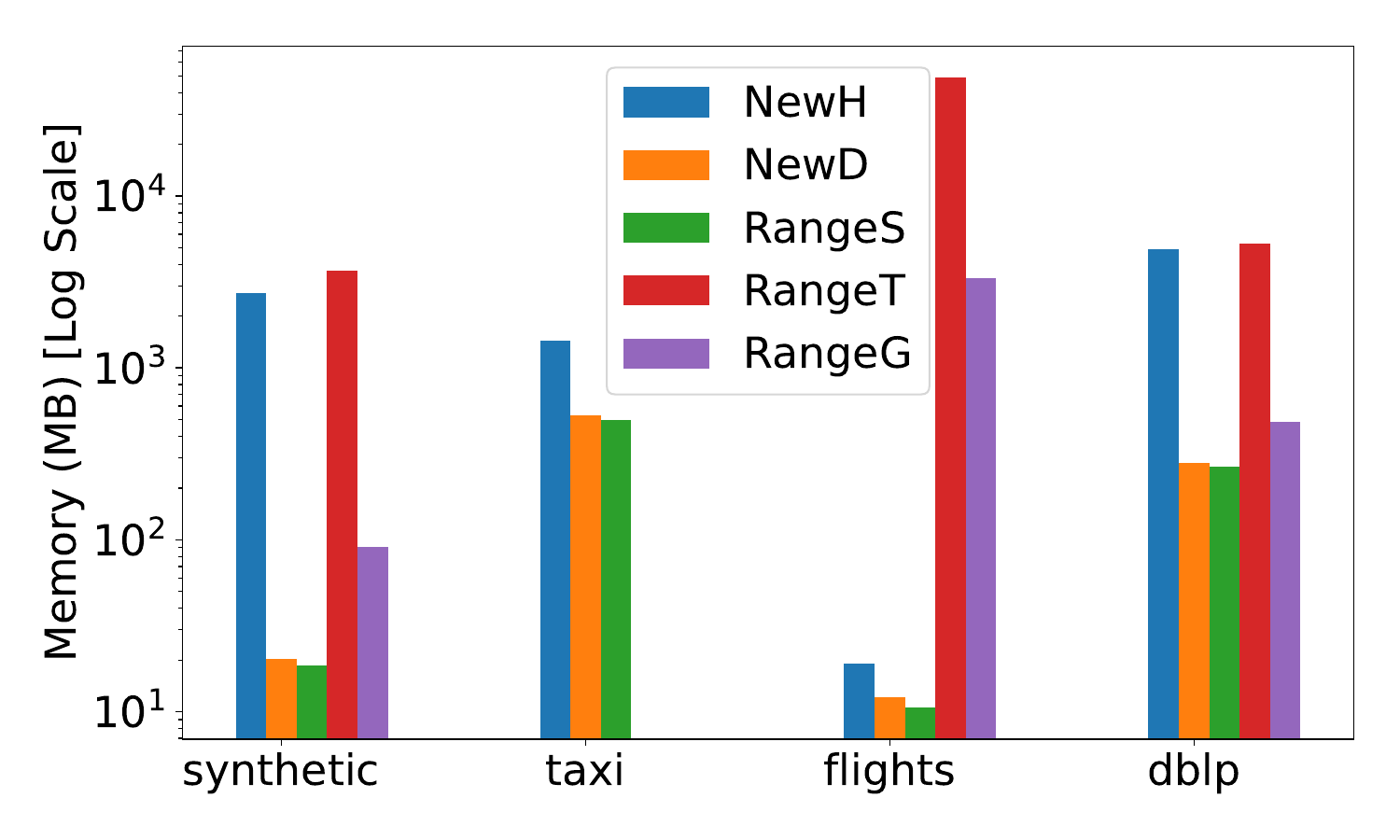}
        \vspace{-2.5em}
		\caption{Memory Usage.} \label{fig:memcomp}
    \end{minipage}
    \hfill
    \begin{minipage}[t]{0.32\textwidth}
		\includegraphics[width=\textwidth]{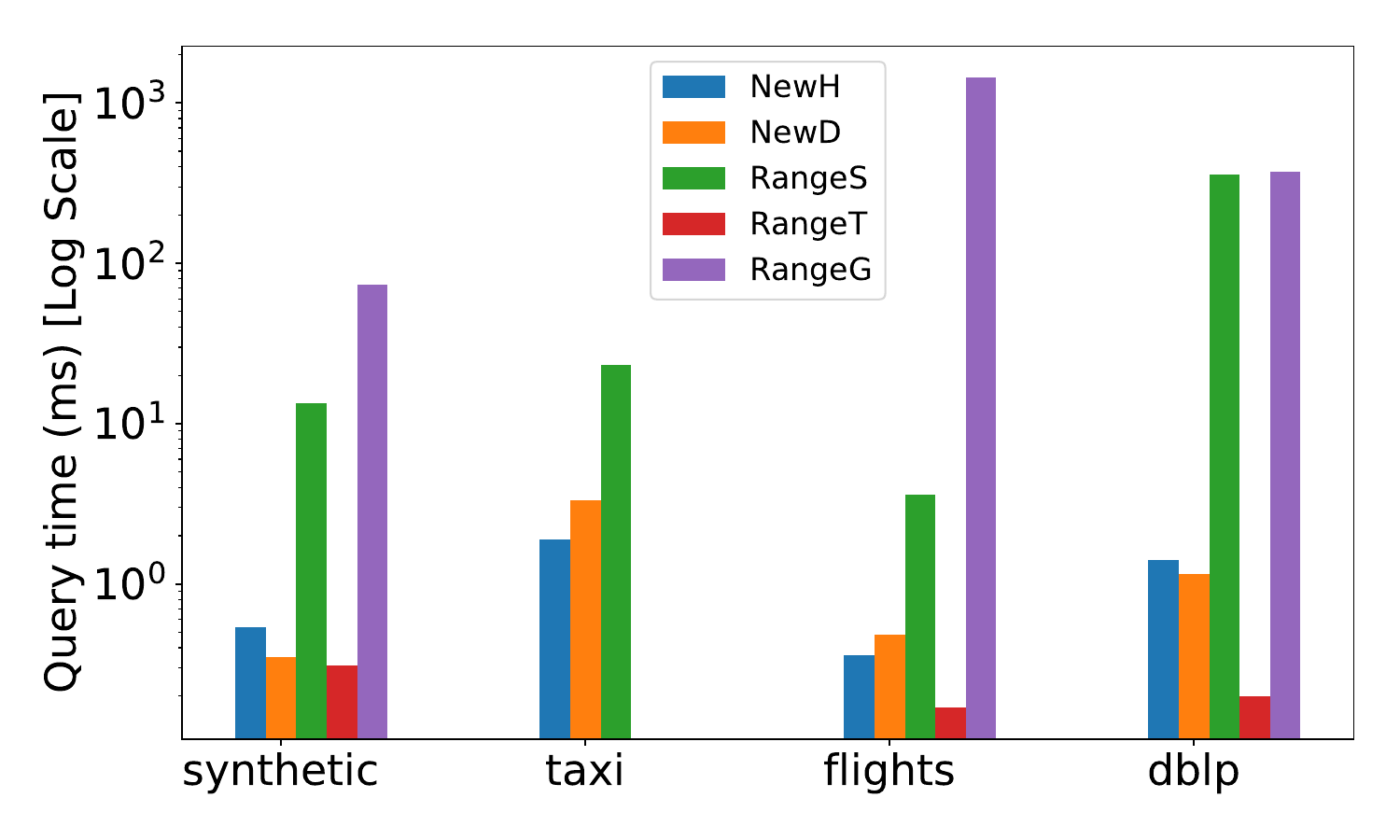}
        \vspace{-2.5em}
		\caption{Query Time.} \label{fig:timecomp}
    \end{minipage}    
    \hfill
    \begin{minipage}[t]{0.32\textwidth}
		\includegraphics[width=\textwidth]{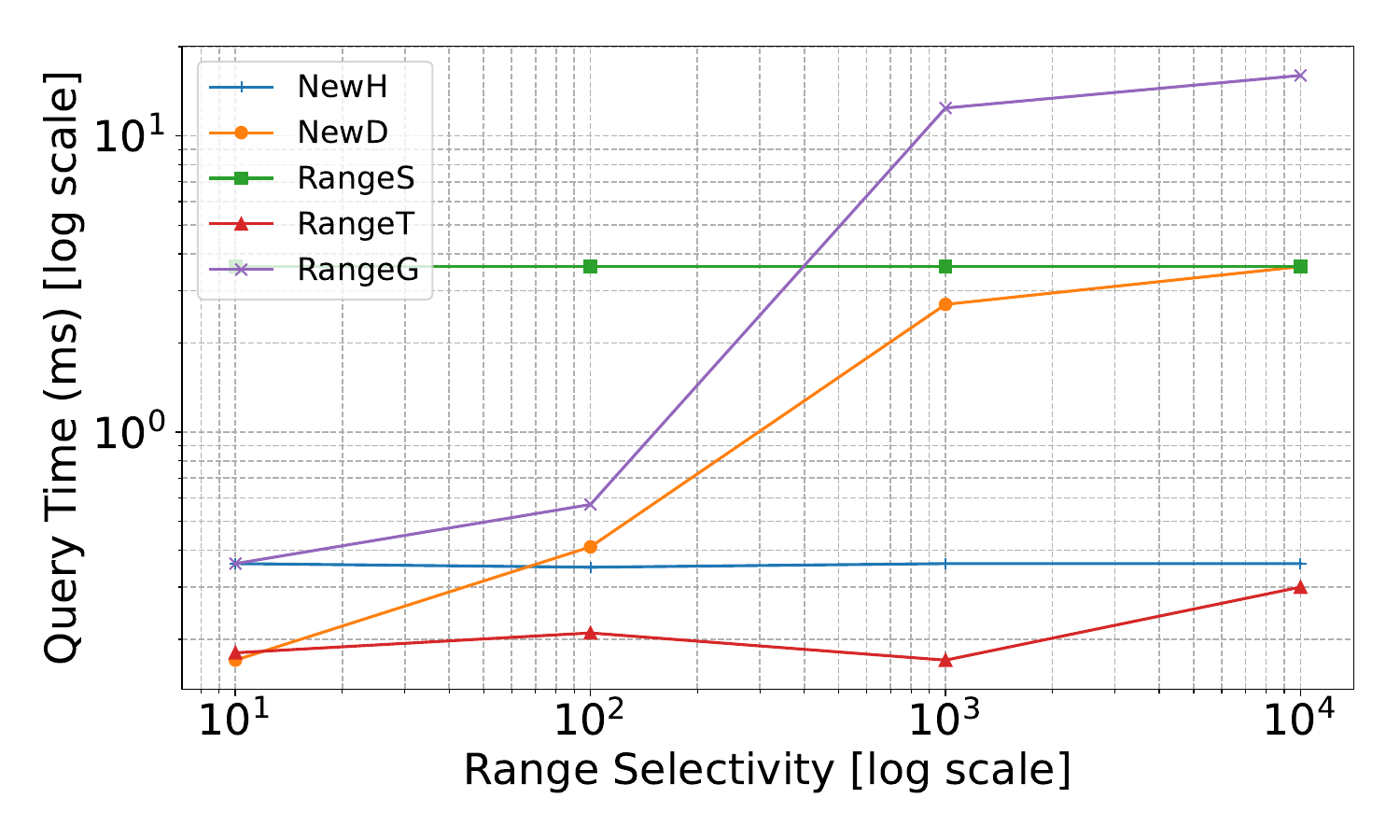}
        \vspace{-2.5em}
		\caption{Impact of Range Selectivity on $\dataflights$ dataset.} \label{fig:size-time-flights}
    \end{minipage}
    \hfill
    \begin{minipage}[t]{0.32\textwidth}
		\includegraphics[width=\textwidth]{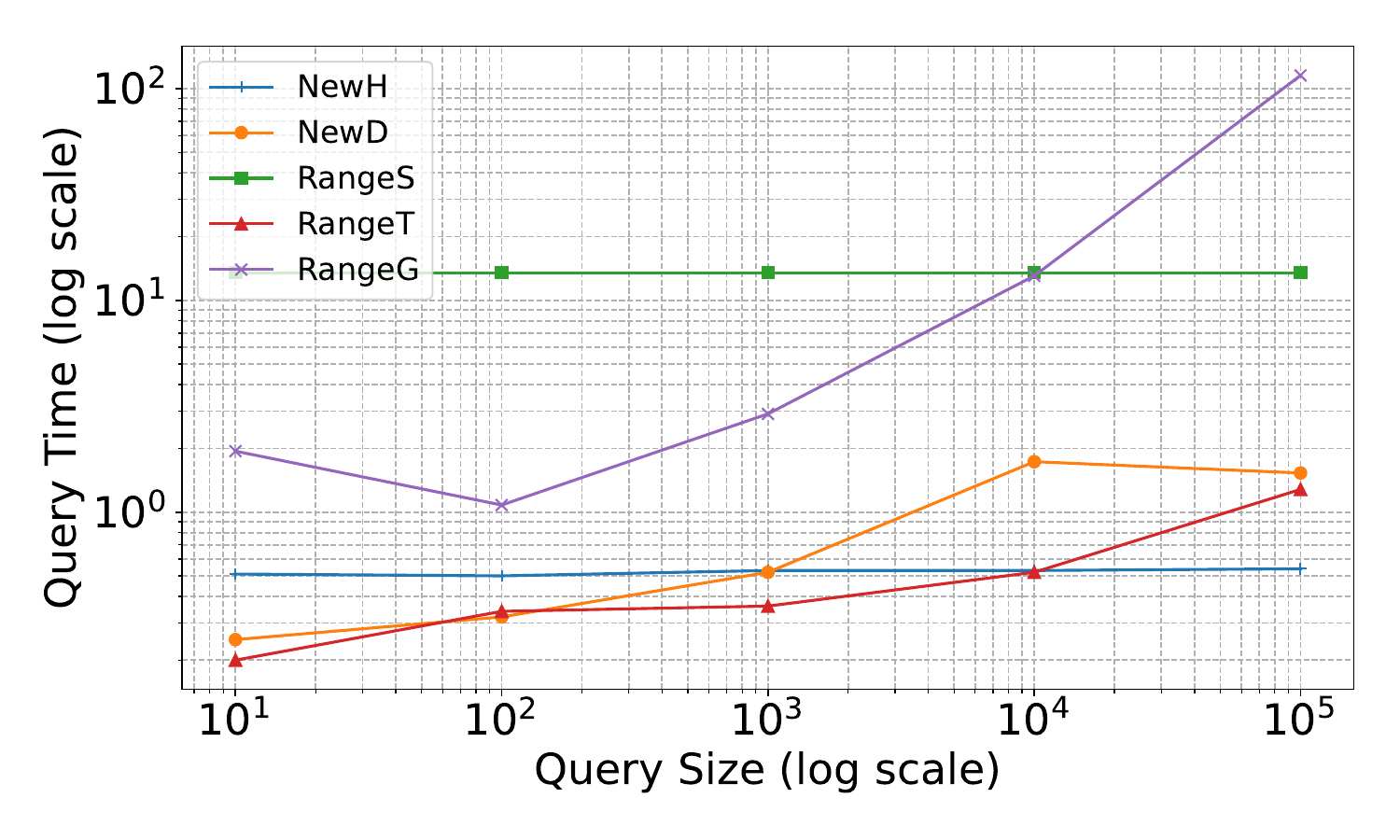}
        \vspace{-2.5em}
		\caption{Impact of Range Selectivity on $\datasyn$ dataset.} \label{fig:size-time-syn}
    \end{minipage}
    \hfill
    \begin{minipage}[t]{0.32\textwidth}
		\includegraphics[width=\textwidth]{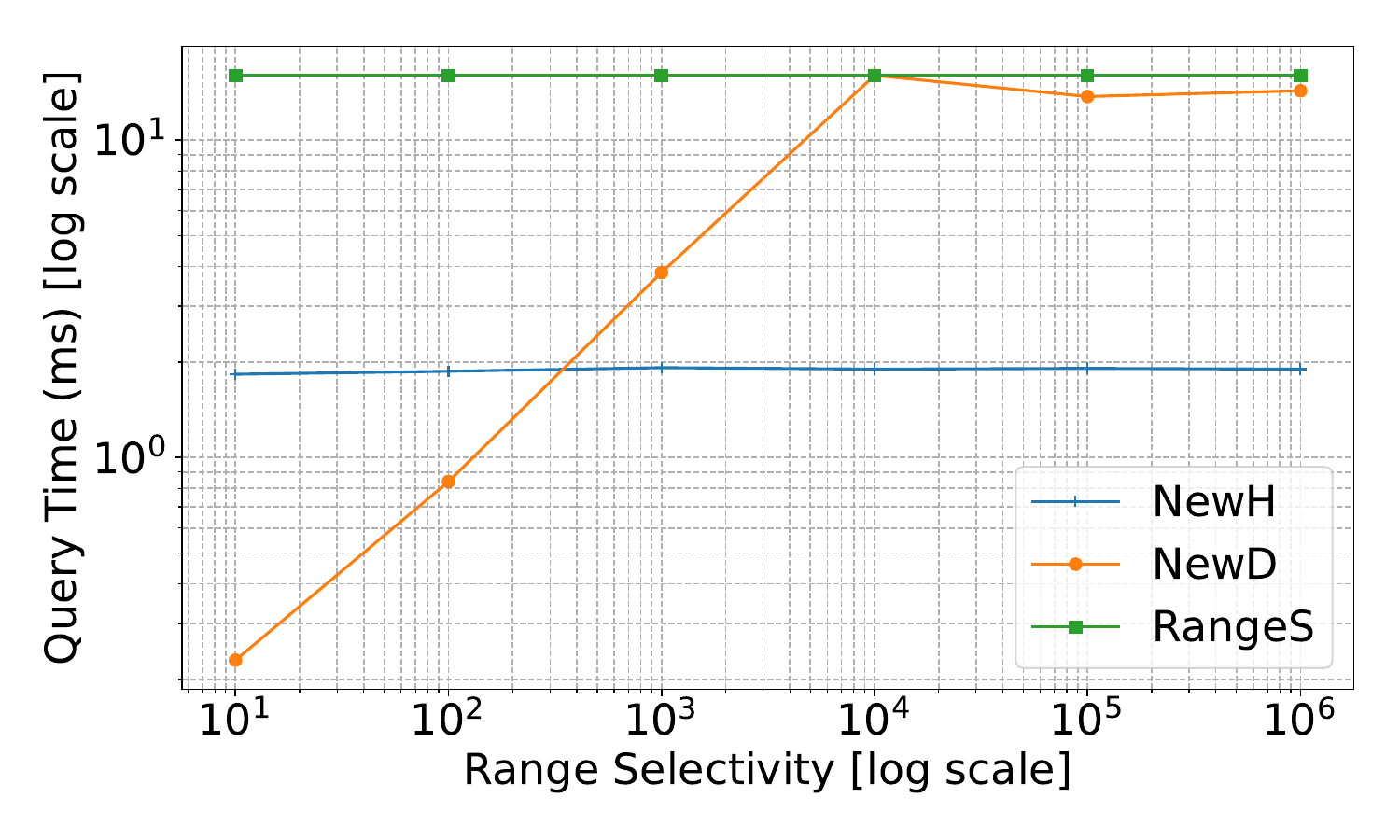}
        \vspace{-2.5em}
		\caption{Impact of Range Selectivity on $\datataxi$ dataset.} \label{fig:size-time-taxi}
    \end{minipage}
    \hfill
    \begin{minipage}[t]{0.32\textwidth}
		\includegraphics[width=\textwidth]{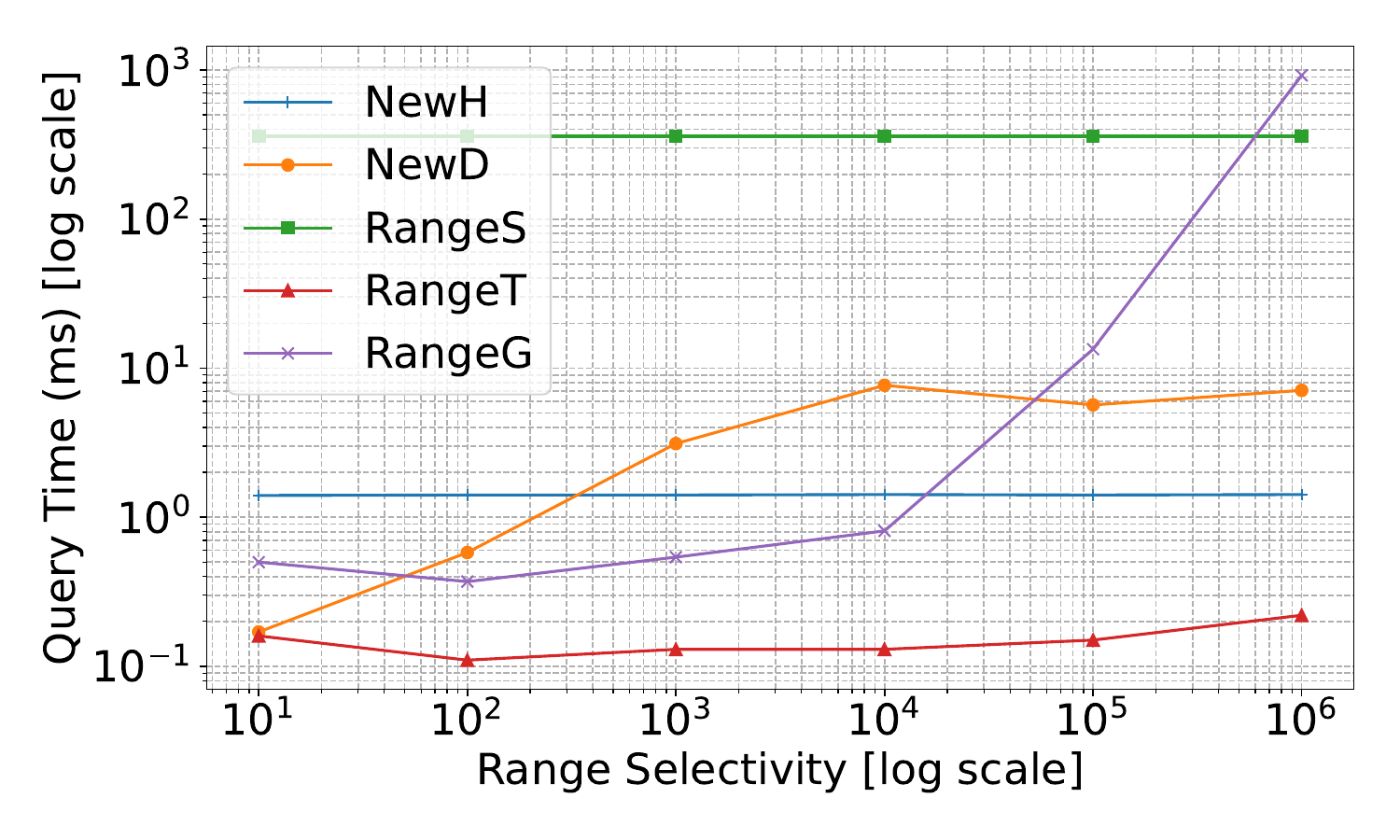}
        \vspace{-2.5em}
		\caption{Impact of Range Selectivity on $\datadblp$ dataset.} \label{fig:size-time-dblp}
    \end{minipage}
 \vspace{-1em}
\end{figure*}

\begin{figure*}[h!t]
	\centering
    \begin{subfigure}[t]{0.24\textwidth}
		\includegraphics[width=\textwidth]{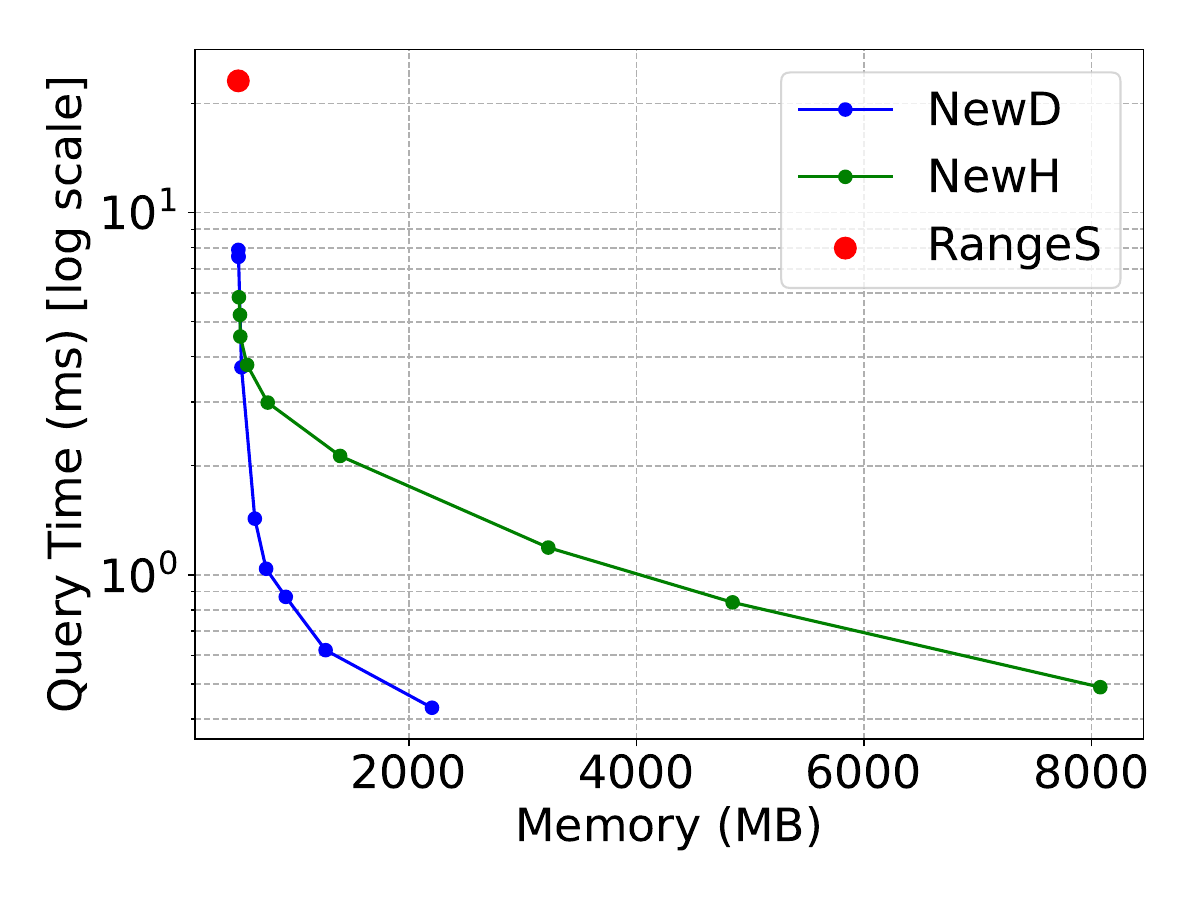}
        \vspace{-2em}
		\caption{$\datataxi$} \label{fig:taxi-trade}
    \end{subfigure}
    \hfill
    \begin{subfigure}[t]{0.24\textwidth}
		\includegraphics[width=\textwidth]{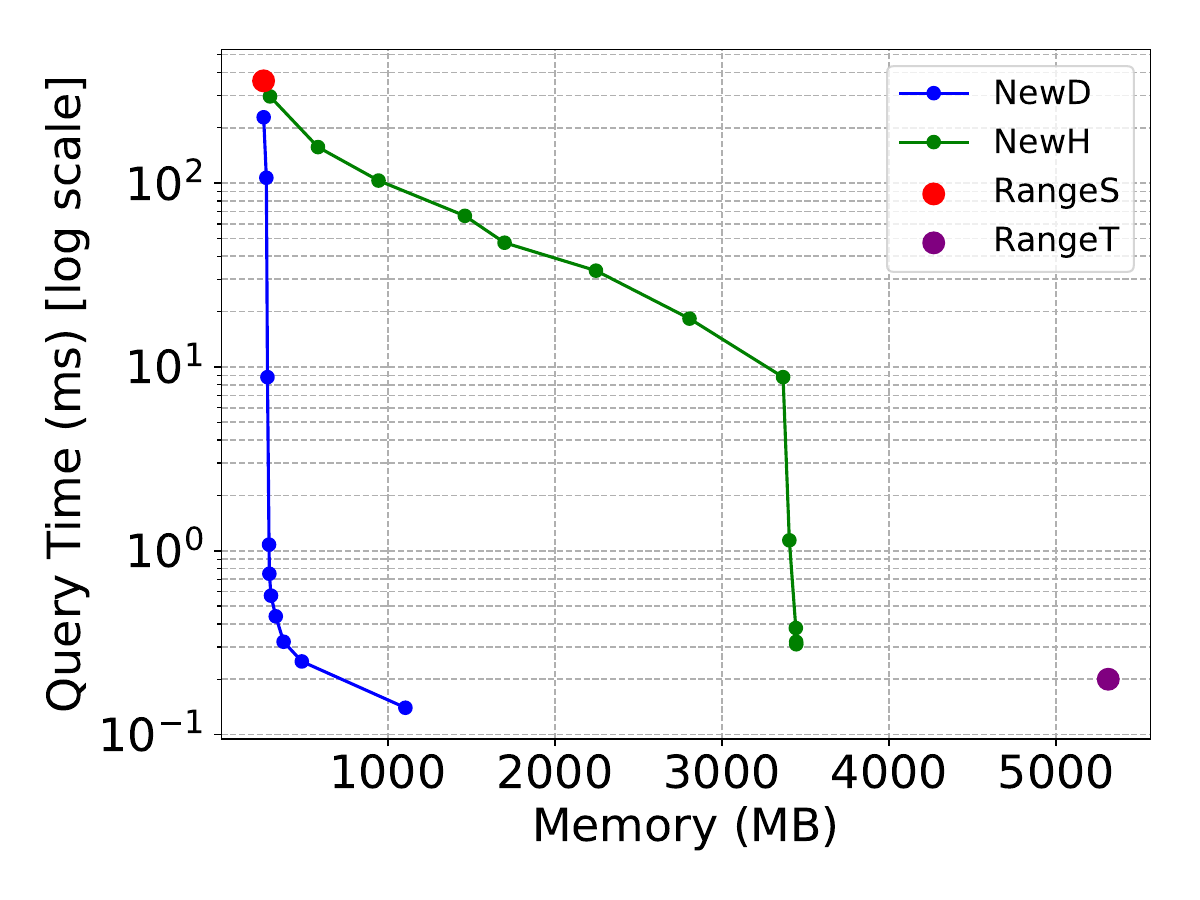}
        \vspace{-2em}
		\caption{$\datadblp$} \label{fig:dblp-trade}
    \end{subfigure}    
    \hfill
    \begin{subfigure}[t]{0.24\textwidth}
		\includegraphics[width=\textwidth]{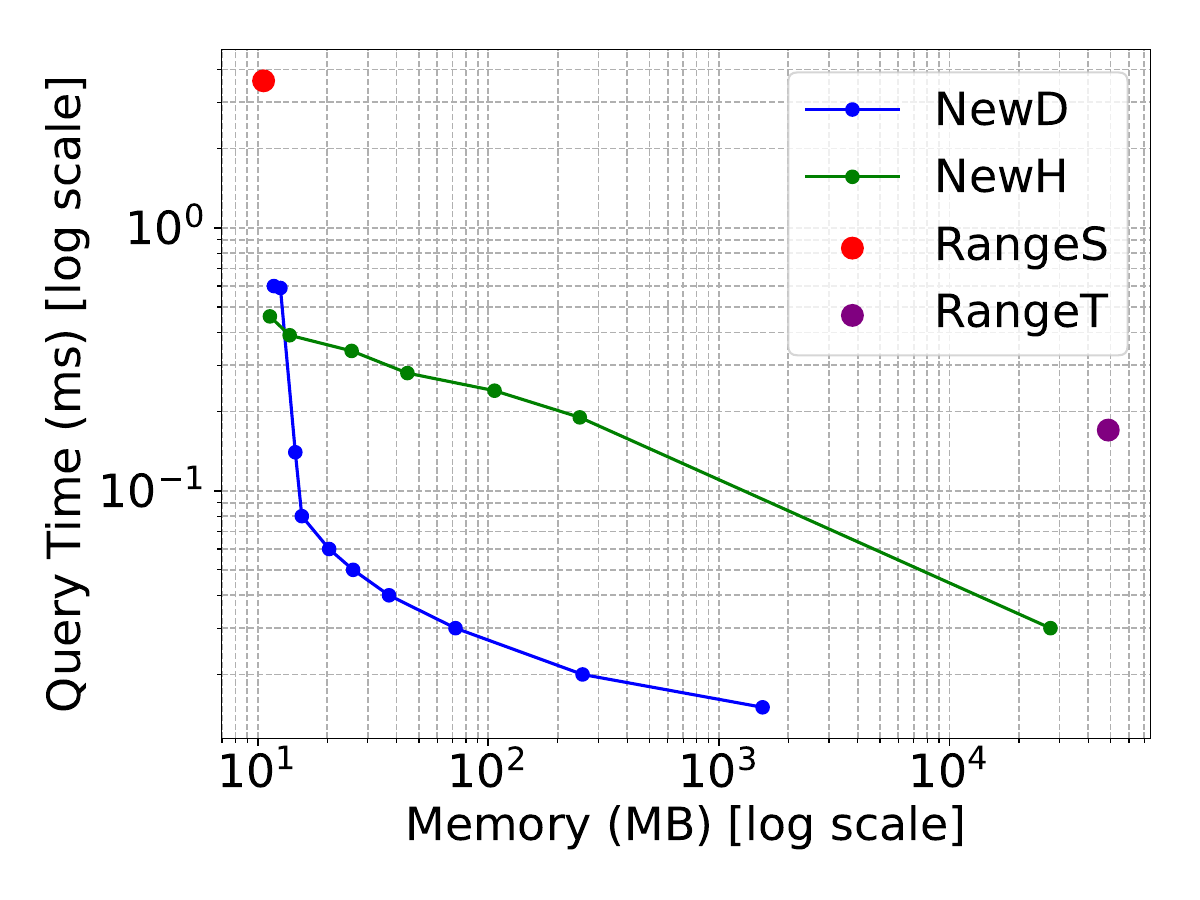}
        \vspace{-2em}
		\caption{$\dataflights$} \label{fig:flight-trade}
    \end{subfigure}
    \hfill
    \begin{subfigure}[t]{0.24\textwidth}
		\includegraphics[width=\textwidth]{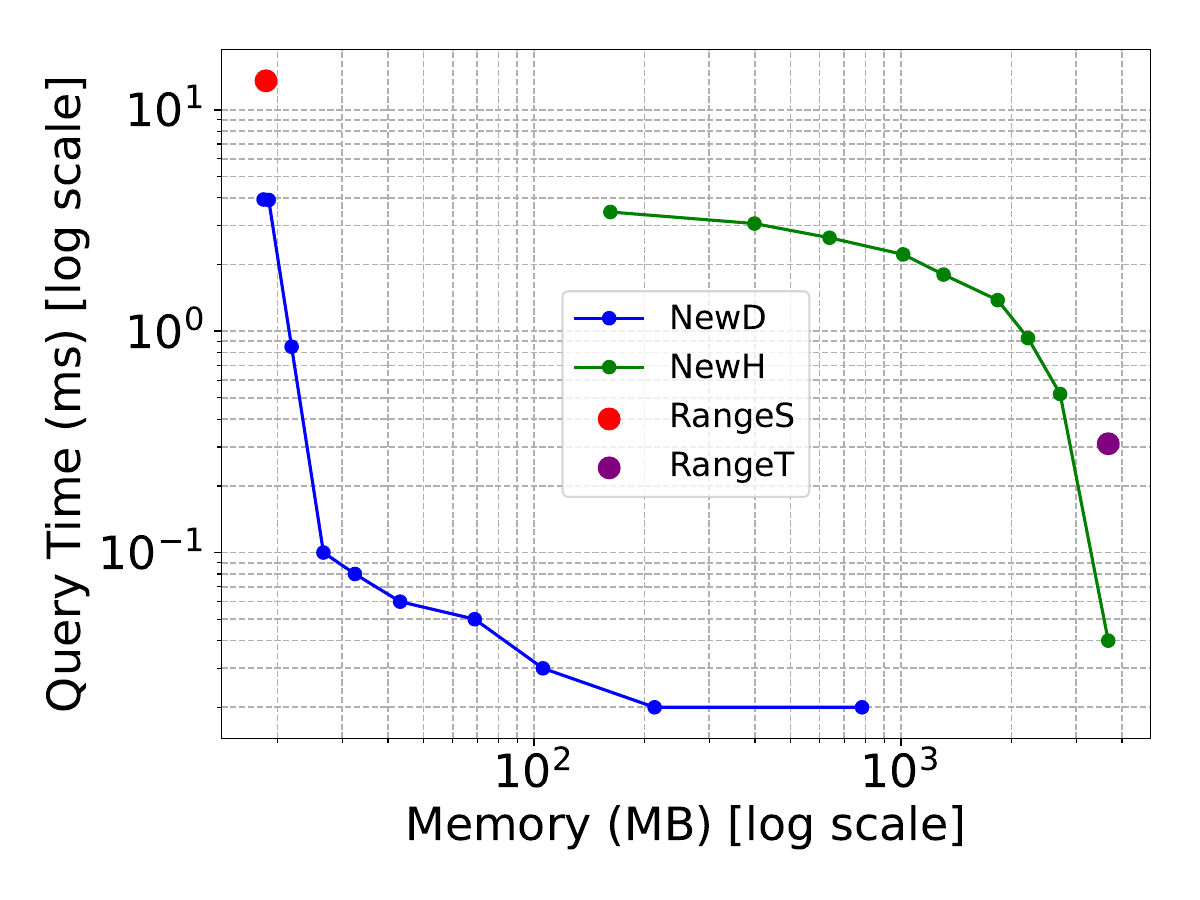}
        \vspace{-2em}
		\caption{$\datasyn$} \label{fig:syn-trade}
    \end{subfigure}
    \hfill
 \vspace{-1em}
 \caption{Memory usage vs Query Time tradeoff over different datasets.}
\end{figure*}

\section{Experiments}
\label{sec:experiments}
\subsection{Setup}
All our experiments were implemented in C++ and performed on a Linux machine with a 3.5 GHz Intel Core i9 processor and 128 GB of memory. All the codes are public at \cite{code}.

\paragraph{Queries and Algorithms} 
We evaluate all algorithms over the matrix query $\Q_\textsf{matrix}(A, C):- R_1(A, B), R_2(C, B)$, where $R_1$ and $R_2$ are relations from different datasets $I$, as described later.  
We use multiple randomly generated rectangles $\rect = [l_{A}, r_{A}] \times [l_{C} \times r_{C}]$ and report the average query time and memory used for answering $\prob(\Q_\textsf{matrix}, I, \rect)$. 
The exact process of generating the rectangles is shown later in Section \ref{sec:expresults}.
We compare our indexes with the following three baselines:
\begin{itemize}[leftmargin=*]
    \item $(\algnjoin)$ This is a version of Yannakakis~\cite{yannakakis1981algorithms} algorithm specifically optimized for $\Q_{\textsf{matrix}}$. For each value $b \in \adom(B)$, we store two sorted lists in the preprocessing time. One containing the list of $A$-values $\proj_{A}(\sigma_{B=b} R_1)$ and one containing the $C$-values $\proj_{C}(\sigma_{B=b} R_2)$, denoted by $L^b_1$ and $L^b_2$, respectively. When a query $\rect = [l_{A}, r_{A}] \times [l_{C} \times r_{C}]$ comes, we go over each $b \in \adom(B)$, and obtain two counters $\Delta^b_1 = |\{a \in L^b_1 : l_{A} \leq a \leq r_{A}\}|$ and $\Delta^b_2 = |\{c \in L^b_2 : l_{C} \leq c \leq r_{C}\}|$ the number of values satisfying the given ranges in $L^b_1$ and $L^b_2$, respectively. Finally, we return $\sum_{b \in \adom(B)}\Delta^b_1 \cdot \Delta^b_2$. For each $b$, we run a binary search on $L^b_1$ and $L^b_2$ to obtain $\Delta^b_1$ and $\Delta^b_2$, respectively.
    \item ($\algranget$) We materialize all the results of the underlying CQ in a range tree index and then use the standard range tree query procedure to answer $\prob$ queries. 
    \item ($\alggist$) We materialize all the results of the underlying CQ in a table and then use the GIST spatial index from the PostGIS extension of PostgreSQL to answer $\prob$ queries. 
\end{itemize}

 We implemented our algorithm for $2$-star queries described in Section \ref{sec:opt}, which almost matches the lower bounds and the more general algorithm described in Section \ref{sec:general}. These algorithms are denoted as $\algdec$ and $\algheavy$, respectively.
For the implementation of $\algdec$, we added the following simple optimization. 
When a query comes up, we can quickly calculate the number of operations we need to make to answer the query. If this number is larger than the number of distinct values in the active domain of $B$, we use the algorithm $\algnjoin$ instead. Using this trick, we can quickly answer the queries in datasets with a rather small number of different values for the attribute $B$. Moreover, as described before, a tradeoff between memory usage and query time for both algorithms $\algdec$ and $\algheavy$ can be tuned as desired. We tuned the parameters for these algorithms in different ways, which will be elaborated later.

\paragraph{Datasets}
We use one synthetic and three real-world datasets: 

$\datasyn$: We generate 10,000 tuples for relation $R_1$ (resp. $R_2$) as follows: a tuple $t$ is produced by setting $\pi_{A}(t)$ (resp. $\pi_{C}(t)$) and $\pi_{B}(t)$ to be drawn uniformly at random from $[100000]$ and $[4500]$, respectively. The intuition is to ensure each table is not very small and the result of the underlying CQ is not too large simultaneously so we can materialize the result of the underlying CQ for the baselines $\algranget$ and $\alggist$. Moreover, to ensure the size of the active domain $\adom(B)$ is large enough, so we can truly compare the performance of $\algheavy$ and $\algnjoin$, since their performance depends on the size of the active domain of $B$.

$\datataxi$\cite{taxi}. The Yellow Taxi Trip records data from NYC Taxi and Limousine Commission data containing around $3$ million rides taken in January 2024 in New York City. We set the attribute $A$ as the drop-off time, attribute $B$ as the fare amount, and $C$ as the pick-up time. We set the relations $R_1$ and $R_2$ to contain the data projection on $A, B$ and $C, B$, respectively. 

$\dataflights$\cite{flights}. The Flights dataset contains the records of 56710 flights in 610 different airports. To fill the $R_1$ relations, we set the attributes $A$ as the arrival time and $B$ as the arriving airport and project the data on these attributes. We fill $R_2$ similarly by setting $C$ as the departure time and $B$ as the departing airport. 

$\datadblp$\cite{dblp}. A collaboration graph is constructed, where two authors are connected if they publish at least one paper. We set $R_1$ and $R_2$ to be the same undirected co-authorship graph, and hence all the attributes $A$, $B$, and $C$ are the authors' IDs. We use this dataset to evaluate the performance of all algorithms on graph data. 
\vspace{-1em}
\begin{table}[h]
    \centering
    \resizebox{0.5\linewidth}{!}{
    \begin{tabular}{|c|c|c|c|c|c|c|c|}  
    \toprule
    {\bf Dataset} & {\bf \#Records} & {\bf $|\Q_\textsf{matrix}(\I)|$} & {\bf $|\adom(B)|$}\\ \hline
    Taxi 
    & 
        2964624
        &
        236862801792
        &
        8970
        \\ \hline
        DBLP
        & 
        1049866
        &
        7064738
        &
        304374
        \\ \hline
        Flights
        & 
        56710
        &
        48096866
        &
        610
        \\ \hline
        Synthetic
        & 
        100000
        &
        2321902
        &
        4500
        \\ \hline
        
    \end{tabular}
    }
    \caption{\label{table:datasets}An overview of the datasets.}
    \vspace{-2em}
\end{table}

\subsection{Experimental Results}\label{sec:expresults}
As mentioned before, the tradeoff between the query time and the memory usage of our algorithms $\algdec$ and $\algheavy$ can be controlled by changing the parameter $T$ as shown in Theorems \ref{alg:star} and \ref{range}. This can be very useful in practical scenarios where the space is limited since none of the baseline algorithms allow this flexibility, and all of them will use a fixed amount of memory based on the database instance. For the first part of the experiments, we fixed the parameter $T$ to be close to $\sqrt{N}$ for both of our algorithms. By theorems \ref{alg:star} and \ref{range}, this results in an asymptotic query time of $\O(\sqrt{N})$ for both the algorithms while using $O(N)$ additional space for the $\algdec$ algorithm and $O(N^{3/2})$ space for the $\algheavy$ algorithm. Note that although $\algheavy$ may use more memory than $\algdec$, it can be used on a more generalized set of queries, as shown in Section \ref{sec:general}. For the $\datataxi$ dataset, we omit the results for algorithms $\algranget$ and $\alggist$, since they need to materialize all the results in $\Q_\textsf{matrix}(\I)$. The huge size of $\Q_\textsf{matrix}(\I)$, as shown in Table \ref{table:datasets}, makes it impossible to store all the results on our machine. Hence, it becomes obvious that the baselines $\algranget$ and $\alggist$ are not practical even in databases with a couple of million tuples.

\paragraph{Memory and Query time}
With these settings, the memory usage of algorithms over all the datasets is shown in Figure \ref{fig:memcomp}. To construct the query ranges $\rect$, from the sorted list of values in the active domain of $A$ (and $C$), we chose two numbers drawn uniformly at random to use as values $l_{A}$ and $r_{A}$ ($l_{C}$ and $r_{C}$) and set $\rect = [l_{A}, r_{A}] \times [l_{C}, r_{C}]$. We generate 100 different ranges and report the average query time to compare the query time performance shown in Figure \ref{fig:timecomp}. We observe that the memory usage of the two indexes $\algdec$ and $\algnjoin$ is very close to each other and significantly less than all the other baselines. This is as expected since these two indexes are the only ones that use space linear to the input size in theory. On the datasets $\datasyn$ and $\datadblp$, we can see that the memory usage of $\algheavy$ is large and is close to the $\algranget$, which has the worst memory consumption. However, for the datasets $\datataxi$ and $\dataflights$, $\algheavy$ shows a much better performance compared to both $\algranget$ and $\alggist$. This is because in the $\dataflights$ and $\datataxi$ datasets, a small number of different values in the active domain of $B$, contribute to the most number of result tuples in $\Q_\textsf{matrix}(\I)$. Intuitively, there are only a few crowded airports, and the other airports have significantly fewer flights passing them, or most taxi rides have a fare in a certain range, and there are very few rides that have a very low or very high fare. More precisely, in $\dataflights$, $4\%$ of different values in the active domain of $B$, contribute to more than $90\%$ of the results in $\Q_\textsf{matrix}(\I)$, and this percentage is $5\%$ for the $\datataxi$ dataset. However, for the $\datadblp$ and $\datasyn$, this percentage is $35\%$ and $80\%$, respectively. Generally, the datasets with a few heavy values in the active domain of $B$, fit better for $\algheavy$, since it can quickly separate them and take care of them at the query time. 
In Figure \ref{fig:timecomp}, we observe that the query time of $\algranget$, beats all the other indexes over all applicable datasets, yet this query performance is not for free since the memory consumption is huge. As we tuned the asymptotic query time complexity of $\algdec$ and $\algheavy$ both equal to $\sqrt{N}$, their performance is close to each other and much better than the baselines except for $\algranget$. It can be seen that $\algdec$ performs significantly better than $\algnjoin$, despite both having similar memory usage. As described above, $\algheavy$ performs better on its preferable datasets $\datataxi$ and $\dataflights$ while $\algdec$ performs better on the other datasets. Although $\alggist$ is one of the most commonly used indexes for spatial range counting, we observe that it answers the queries slower than all the other baselines over all the datasets. We note that this index is not designed to answer range counting over Conjunctive Queries. The underlying $R$-tree index implemented in this index makes it super fast to answer $\prob(\Q_\textsf{matrix}, \I,\rect)$ queries where $\rect$ has small volume or when it contains a small number of tuples, but it performs significantly worse in average without any restriction on $\rect$. 

\paragraph{Range Selectivity}
Let the \textit{selectivity} of range $[l, r]$ over an attribute $A$ and a relation $R$ be the number of tuples $t \in R$, such that $l \leq \pi_{A}(t) \leq r$. We use this definition to evaluate the effect of the query rectangle $\rect$ on the performance of the algorithms. For a fixed selectivity $s$, we randomly take a range $[l_{A}, r_{A}]$ ($[l_{C}, r_{C}]$) such that its selectivity over $A$ ($C$) and $R_1$ ($R_2$) is equal to $s$. Then we set $\rect = [l_{A}, r_{A}] \times [l_{C}, r_{C}]$ and ask the algorithms to compute $\prob(\Q_\textsf{matrix},\I,\rect)$. We repeat this process 100 times and report the average time in Figures \ref{fig:size-time-flights}, \ref{fig:size-time-syn}, \ref{fig:size-time-taxi}, and \ref{fig:size-time-dblp}. It can be seen that the selectivity of $\rect$ has minimal impact on the query performance of $\algheavy$ and $\algranget$, as expected from their construction. On the other hand, $\alggist$ performs significantly better when the selectivity of $\rect$ is smaller. For the smaller selectivities, $\algdec$ performs much better than $\algnjoin$ over all the datasets. For larger selectivities, the query time gets closer, but $\algnjoin$ never beats $\algdec$. In all datasets, $\algdec$ performs better than $\algheavy$ on the smaller selectivities, while $\algheavy$ performs better on larger selectivities. 

\paragraph{Tradeoff}
We can tune the query time of $\algdec$ and $\algheavy$, to meet the constraints on the available memory in practical use cases. The query time and memory usage are fixed for all the other baselines and depend solely on the input dataset. The tradeoff between the query time and memory usage for our algorithms is shown in Figures \ref{fig:taxi-trade}, \ref{fig:dblp-trade}, \ref{fig:flight-trade}, and \ref{fig:syn-trade}. It can be seen that for a fixed amount of memory, on average $\algdec$ outperforms $\algheavy$. However, note that the $\algheavy$ algorithm can be used on a more generalized set of conjunctive queries. Generally, $\algnjoin$ has a low memory consumption but suffers from a large query time. On the other extreme, $\algranget$ can answer the queries quickly, but it suffers from a large memory usage. Interestingly, we observe that by allowing a slight increase in memory usage, we can significantly improve the query performance over all the datasets by using $\algdec$. This improvement can be as large as two orders of magnitude for three of the datasets, as it is shown in Figures \ref{fig:dblp-trade}, \ref{fig:flight-trade}, and \ref{fig:syn-trade}.

\paragraph{Summary}
The main strength of $\algheavy$ and $\algdec$ is their ability to tune the query time based on the available memory. Indeed, when memory usage is out of concern, the best approach is to materialize all the results in $\Q_\textsf{matrix}(\I)$ in an index designed explicitly for range counting queries such as range trees. However, when there is a memory limit, the goal is to answer the counting queries as fast as possible. The experiments show that our indexes can achieve this goal by quickly answering $\prob$ queries while meeting the memory constraint. Neither of our two indexes dominates the other one, and the performance highly depends on the amount of available memory, the input dataset, and the selectivity of the queried ranges.  
Moreover, we showed that when the available space is nearly equal to the size of the dataset, our index $\algdec$ achieves significantly faster query time compared to the baseline $\algnjoin$, while using almost the same amount of memory.

\end{document}